\documentclass[journal]{IEEEtran}
\PassOptionsToPackage{ruled, linesnumbered}{algorithm2e}
\usepackage[T1]{fontenc}
\usepackage[utf8]{inputenc}
\usepackage{color}
\usepackage{mathrsfs}
\usepackage{bm}
\usepackage{algorithm2e}
\usepackage{amsmath}
\usepackage{amsthm}
\usepackage{amssymb}
\usepackage{graphicx}
\usepackage{cite}
\usepackage{balance}
\usepackage{paralist}
\usepackage[bookmarks=true,bookmarksnumbered=true,bookmarksopen=true,bookmarksopenlevel=1,
 breaklinks=false,pdfborder={0 0 0},pdfborderstyle={},backref=false,colorlinks=false]
 {hyperref}
\hypersetup{pdftitle={Your Title},
 pdfauthor={Your Name},
 pdfpagelayout=OneColumn, pdfnewwindow=true, pdfstartview=XYZ, plainpages=false}

\makeatletter
\theoremstyle{plain}
\newtheorem{thm}{\protect\theoremname}
\theoremstyle{proposition}
\newtheorem{prop}{\protect\propositionname}
\theoremstyle{remark}
\newtheorem{rem}[thm]{\protect\remarkname}

\usepackage[caption=false,font=footnotesize]{subfig}
\usepackage{hyperref}
\hypersetup{
    colorlinks=true,
    linkcolor=blue,
    filecolor=magenta,      
    urlcolor=blue,
	citecolor=blue,
    }

\usepackage{amsmath, amsfonts, amssymb, graphicx}
\newcommand{\op}[1]{\operatorname{#1}}
\newcommand{\herm}{^{\mathsf{H}}}
\newcommand{\trans}{^{\mathsf{T}}}
\newcommand{\conj}{^{*}}

\DeclareMathOperator{\tr}{\mathsf{tr}}
\DeclareMathOperator{\vecd}{\mathsf{vec_d}}

\DeclareMathOperator{\diag}{\mathsf{diag}}
\DeclareMathOperator{\maximize}{maximize}
\DeclareMathOperator{\st}{subject~to}
\DeclareMathOperator{\argmin}{argmin}

\DeclareMathOperator{\cov}{\mathsf{Cov}}
\DeclareMathOperator{\vecop}{\mathsf{vec}}
\DeclareMathOperator{\unvec}{\mathsf{unvec}}

\allowdisplaybreaks

\makeatother

\providecommand{\remarkname}{Remark}
\providecommand{\theoremname}{Theorem}
\providecommand{\propositionname}{Proposition}

\begin{document}
\title{A Novel Framework for Transmitter Privacy in Integrated Sensing and Communication}
\author{Vaibhav~Kumar,~\IEEEmembership{Member,~IEEE,} Ahmad~Bazzi,~\IEEEmembership{Senior~Member,~IEEE,} \\
Christina~Pöpper,~\IEEEmembership{Senior~Member,~IEEE,} and Marwa~Chafii,~\IEEEmembership{Senior~Member,~IEEE}\thanks{Vaibhav Kumar, Ahmad Bazzi, and Marwa Chafii are with the Wireless Research Lab, Engineering Division, New York University Abu Dhabi (NYUAD), UAE. Ahmad Bazzi and Marwa Chafii are also with NYU WIRELESS, NYU Tandon School of Engineering, NewYork, USA (e-mail: vaibhav.kumar@ieee.org; ahmad.bazzi@nyu.edu; marwa.chafii@nyu.edu).\protect \\
Christina Pöpper is heading the Cyber Security \& Privacy (CSP) Lab, Center for Cyber Security, Science Division, New York University Abu Dhabi (NYUAD), UAE (e-mail: christina.poepper@nyu.edu).}\thanks{This work was supported by the Center for Cyber Security through New York University Abu Dhabi Research Institute under Award G1104. The work of Marwa Chafii was also supported in part by Tamkeen under the Research Institute NYUAD grant CG017.}}

\maketitle

\begin{abstract}
	Integrated sensing and communication (ISAC) systems introduce new privacy risks, since an unintended sensing node may exploit the shared radio waveform to infer transmitter-related information even when the communication payload itself remains secure. In this paper, transmitter privacy, defined here as limiting unauthorized inference of transmitter-related information through channel estimation, is investigated in a reconfigurable intelligent surface (RIS)-aided multi-antenna wireless system comprising a transmitter, a legitimate receiver, a malicious sensor, and a RIS. The malicious sensor is assumed to estimate the transmitter--sensor channel, and the acquired channel state information may subsequently be used for unauthorized sensing, inference, or related signal-processing tasks. To counter this threat, a privacy-oriented design is considered in which the transmitter employs a superposition-based signaling strategy combining a message-bearing signal and transmit-side artificial noise (AN), while the RIS is used to shape the propagation environment in a privacy-aware manner. The channel-estimation performance at the malicious sensor is first characterized under imperfect prior knowledge, and both the true and predicted mean-square-error expressions are derived. Based on this characterization, a joint active--passive beamforming design problem is formulated to maximize the malicious sensor's predicted channel-estimation error subject to a communication quality-of-service requirement, a transmit-power budget, and the unit-modulus constraints of the RIS. The resulting non-convex problem is addressed via a numerically efficient alternating optimization framework based on an augmented Lagrangian reformulation. Numerical results show that RIS-assisted propagation shaping can significantly impair unauthorized channel estimation compared to the non-RIS counterpart while maintaining reliable communication, and further reveal that the resulting privacy gains extend to a more direct sensing metric, namely the angle-of-arrival (AoA) estimation accuracy at the malicious sensor.
\end{abstract}


\IEEEpeerreviewmaketitle{}

\section{Introduction}

\IEEEPARstart{I}{ntegrated} sensing and communication (ISAC) is emerging as an important architectural direction for future wireless networks, as it enables both data transmission and environmental sensing over a common radio platform~\cite{26_COMST_ISAC_over_the_years, 23_COMST_Marwa, 26_TWC_Ling}. Through the shared use of spectrum, infrastructure, hardware chain, and signaling resources, ISAC can improve spectrum utilization, reduce deployment cost, and support a broad range of emerging services, including localization and tracking, autonomous driving, industrial automation, smart environments, digital twins, and immersive extended-reality applications~\cite{25_OJCOMS_MagboolSurvey}. These advantages have attracted significant interest from academia and industry, and have positioned ISAC as a realistic design paradigm for next-generation wireless systems.

This momentum is also reflected in the ongoing 3GPP Release~19 effort, where ISAC has entered formal study activities through TR~22.837~\cite{3gpp_tr22837}, service-level specification efforts through TS~22.137~\cite{3gpp_ts22137}, and channel modeling efforts in TR~38.901~\cite{3gpp_tr38901}. More broadly, this indicates that future wireless networks may no longer treat sensing as a secondary function supported indirectly by communication, but instead as a native capability embedded into the network architecture. At the same time, however, such close integration also brings new concerns. When communication and sensing share the same radio interface, the transmitted and reflected signals may reveal much more information about users, devices, and the surrounding environment than in conventional communication-only systems. Consequently, in addition to the performance gains enabled by ISAC, it becomes necessary to understand the security and privacy risks created by this tight coupling between communication and environmental inference.

Recent research has begun to address the security challenges of ISAC from a physical-layer perspective~\cite{24_IoTM_ISAC_Priv_Secure,25_COMST_Secure_ISAC_Survey}. A first line of work focuses on secrecy-oriented secure ISAC design, including secure beamforming~\cite{24_TWC_SecureBeamforming-I, 25_TVT_SecureBeamforming-II, 26_JSAC_Comm3D_Sense2D_SecureBeamforming-III}, artificial-noise (AN) injection~\cite{23_TCOM_AN-I, 24_TIFS_AN-II, 25_TCOM_AN-III}, and secrecy-rate maximization~\cite{22_TCOM_Secure_ISAC_SecrecyRateMaximization-I, 26_JSAC_SecrecyRateMaximization-II}, where the objective is to protect confidential communication while maintaining the desired sensing and communication functionalities. The aforementioned direction has also been extended to reconfigurable intelligent metasurface (RIS)-assisted settings, where the programmable propagation environment is jointly optimized with the transmission design to enhance communication confidentiality~\cite{25_IoTJ_Secure_RIS_ISAC-I, 25_TWC_Kumar_Secure_RIS_ISAC-II, 25_ICC_Ainara_Secure_RIS_ISAC-III,	26_TWC_Secure_RIS_ISAC-IV, 26_JSAC_Secure_RIS_ISAC-V}. In parallel, covert ISAC has also gained significant attention, where the objective is not merely to secure the message content but to conceal the very existence of communication from a warden~\cite{25_TWC_Covert-I, 25_TWC_Covert-II, 26_JSAC_Covert-III}.

Compared to conventional communication-only systems, security in ISAC has a fundamentally broader meaning, since it involves not only \emph{communication-centric} security but also \emph{sensing-centric} security. Communication-centric security follows the classical physical-layer security paradigm, where the goal is to prevent an unintended receiver from decoding the confidential message. In ISAC, communication-centric security can be broadly achieved in two ways. When the eavesdropper is spatially separated from the sensing target, the transmitter may steer nulls toward the suspicious direction to degrade the received signal quality at the eavesdropper. When the eavesdropper coincides with the sensing target, the transmitter may instead inject strong interference or AN so as to impair message decoding at that node. However, such approaches are not sufficient for sensing-centric security, where the malicious node may not be interested in the communication message at all, but rather in performing channel estimation, parameter estimation, detection, localization, or other sensing-related inference tasks. In this case, high interference or noise does not necessarily eliminate the sensing threat, and steering a null toward the malicious sensor may itself be infeasible, particularly when the transmitter has only a small number of antennas, as is typical for user devices. Therefore, unlike conventional physical-layer security, sensing-centric security in ISAC cannot be characterized solely through communication secrecy, and instead calls for dedicated designs that explicitly limit unauthorized sensing and inference.

Encouragingly, a small but growing body of work have begun to move beyond communication-centric secrecy and toward sensing-centric security in ISAC. In particular, RIS-assisted target-obfuscation designs have been developed to protect a sensing region from an adversarial detector by jointly optimizing the transmit beamformer, RIS configuration, and RIS meta-atom assignment so as to reduce the detector's sensing capability while preserving the intended communication and sensing functions~\cite{26_TWC_Magbool_TargetObfuscation}. In a related direction, transmitter location privacy has been investigated through beam-pattern obfuscation, where the angular power distribution observed by a sensing-capable receiver is deliberately reshaped so that a false direction appears dominant without nulling the line-of-sight component, thereby showing that privacy may be achieved through controlled obfuscation rather than pure suppression~\cite{26_ICC_Kundu}. More recently, sensing-secure ISAC signaling has been proposed via \emph{ambiguity-function engineering}, where artificial ghost targets are introduced into the unauthorized observer's range profile to degrade its sensing performance without requiring prior knowledge of the observer's channel state information (CSI)~\cite{26_TWC_AmbiguityFunctionEngineering}.


These works constitute important first steps toward sensing-centric security, since they explicitly recognize that unauthorized sensing may remain possible even when communication secrecy is preserved. Here, sensing-centric security refers to the broader objective of preventing, limiting, or degrading malicious sensing capabilities in ISAC systems, whereas sensing privacy more specifically concerns protecting sensitive information about legitimate entities from being inferred through sensing. In this paper, we use the term \emph{transmitter privacy} to denote the ability of a transmitter to communicate while limiting an unauthorized sensing node's ability to infer transmitter-related information from the shared waveform, particularly through estimation of the transmitter--sensor channel and the downstream inference tasks enabled by that estimate, such as localization, tracking, or AoA inference. Under this definition, transmitter privacy is broader than transmitter location privacy, which focuses only on concealing the transmitter position, and is also distinct from target obfuscation or waveform-level sensing degradation, which primarily protect the sensed target or disrupt a particular sensing output. By contrast, the setting of interest here focuses on malicious channel estimation as a fundamental leakage mechanism, since an estimate of the transmitter--sensor channel can serve as an enabler for subsequent coherent processing, parameter inference, and environment-aware sensing tasks.

Motivated by the above observations, this paper studies \emph{RIS-aided transmitter privacy in a multi-antenna wireless system} with a legitimate receiver and a malicious sensing node. The malicious node is assumed to estimate transmitter-related channel information from the received waveform and use it for unauthorized sensing or inference. To impair such malicious estimation, we propose a privacy-oriented joint design of the transmit precoder and RIS reflection coefficients, where AN and RIS-assisted propagation shaping are exploited to degrade the sensing capability of the malicious node while guaranteeing reliable communication at the legitimate receiver. The resulting design problem is non-convex due to the coupling between active transmission and passive propagation control. An efficient alternating optimization (AO)-based approach is developed to solve it. The results show that the proposed RIS-aided design achieves clear transmitter-privacy gains compared with a non-RIS counterpart.

The main contributions of this work are summarized as follows:
\begin{compactenum}
    \item A RIS-aided framework for transmitter privacy is developed for a wireless system with a legitimate receiver and a malicious sensing node, where transmitter privacy is defined as limiting the unauthorized inference of transmitter-related information through malicious estimation of the transmitter--sensor channel. Within this framework, malicious channel estimation is treated as a fundamental leakage mechanism that can enable subsequent sensing and inference tasks, and transmitter privacy is operationalized through the deliberate degradation of such unauthorized channel estimation.
	\item A privacy-oriented joint design problem is formulated for the transmit precoder and RIS reflection coefficients, with the objective of maximizing the malicious sensor's predicted channel-estimation error while ensuring reliable communication at the legitimate receiver under a transmit-power budget and the unit-modulus constraints of the RIS.
	
	\item To solve the resulting non-convex problem, an efficient AO-based algorithm is proposed for the joint design of active transmit signaling and passive RIS control, and its convergence and computational complexity are analyzed.
	
	\item Numerical results demonstrate that RIS-assisted propagation shaping can significantly degrade unauthorized channel estimation relative to non-RIS benchmarks while preserving the desired communication performance, further illustrate the effects of key system parameters such as the RIS size, antenna configuration, observation length, transmit power, and prior mismatch, and show that the resulting privacy gains also manifest in degraded angle-of-arrival (AoA) estimation at the malicious sensor.
\end{compactenum}

The remainder of this paper is organized as follows. Section~\ref{sec:System-Model-and-Problem-Formulation} presents the considered system model and introduces the channel-estimation framework at the malicious sensor, along with the adopted privacy metric and problem formulation. Section~\ref{sec:AO-Based-Proposed-Solution} develops the proposed AO-based design for the joint optimization of the transmit precoder and RIS reflection coefficients. Section~\ref{sec:Results-and-Discussion} provides numerical results to evaluate the privacy performance of the proposed scheme and to illustrate the impact of key system parameters. Finally, Section~\ref{sec:Conclusion} concludes the paper.

\paragraph*{Notation}

Bold uppercase and lowercase letters denote matrices and vectors, respectively. The set of all $M\times N$ complex-valued matrices is denoted by $\mathbb{C}^{M\times N}$, and the set of positive real numbers is denoted by $\mathbb R_+$. For a matrix $\mathbf{X}$, the transpose, Hermitian transpose, trace, determinant, Frobenius norm, and expectation are denoted by $\mathbf{X}\trans$, $\mathbf{X}\herm$, $\tr(\mathbf{X})$, $\det(\mathbf{X})$, $\|\mathbf{X}\|_{\mathsf{F}}$, and $\mathbb{E}\{\mathbf{X}\}$, respectively. The notation $\|\cdot\|$ denotes the Euclidean norm for vectors. The operator $\diag(\mathbf{x})$ returns a diagonal matrix whose main diagonal is formed by the entries of $\mathbf{x}$, $\vecd(\mathbf{X})$ denotes the column vector formed by the diagonal entries of $\mathbf{X}$, $\vecop(\mathbf{X})$ denotes the column vector obtained by stacking the columns of $\mathbf{X}$, $\unvec_{M\times N}(\mathbf{x})$ denotes the inverse operation of $\vecop(\cdot)$ that reshapes $\mathbf{x}$ into an $M\times N$ matrix, and $\mathbf{X}\otimes\mathbf{Y}$ denotes the Kronecker product of $\mathbf{X}$ and $\mathbf{Y}$. For a real-valued function $f(\cdot)$, the gradient with respect to a complex-valued variable $\mathbf{X}$ is defined as $\nabla_{\mathbf{X}} f(\cdot)\triangleq \frac{\partial f(\cdot)}{\partial \mathbf{X}^{*}}=\frac{1}{2}\left[\frac{\partial f(\cdot)}{\partial \Re\{\mathbf{X}\}}+j\frac{\partial f(\cdot)}{\partial \Im\{\mathbf{X}\}}\right]$, where $\mathbf{X}^{*}$ denotes the complex conjugate of $\mathbf{X}$, and $\Re\{\mathbf{X}\}$ and $\Im\{\mathbf{X}\}$ denote its real and imaginary parts, respectively. The operator $\op{d}(\cdot)$ denotes the differential of its argument. The Euclidean projection of $\mathbf{x}$ onto a feasible set $\mathcal{X}$ is defined as $\Pi_{\mathcal{X}}(\mathbf{x})\triangleq\argmin_{\tilde{\mathbf{x}}\in\mathcal{X}}\|\tilde{\mathbf{x}}-\mathbf{x}\|_2$.

\section{System Model and Problem Formulation\label{sec:System-Model-and-Problem-Formulation}}


We consider the RIS-aided communication system shown in Fig.~\ref{fig:sysMod}, which consists of a transmitter (A), a legitimate receiver (B), a malicious sensor (S), and an RIS (R).\footnote{For analytical tractability, this work restricts attention to a single legitimate receiver and a single malicious sensor; extensions to multi-user and multi-attacker settings are left for future work.} The numbers of antennas at A, B, and S are denoted by $m_{\mathrm{A}}$, $m_{\mathrm{B}}$, and $m_{\mathrm{S}}$, respectively, while the RIS is equipped with $m_{\mathrm{R}}$ reflecting meta-atoms. The wireless channels corresponding to the A--B, A--S, A--R, R--B, and R--S links are denoted by $\mathbf{H}_{\mathrm{AB}}\in\mathbb{C}^{m_{\mathrm{B}}\times m_{\mathrm{A}}}$, $\mathbf{H}_{\mathrm{AS}}\in\mathbb{C}^{m_{\mathrm{S}}\times m_{\mathrm{A}}}$, $\mathbf{H}_{\mathrm{AR}}\in\mathbb{C}^{m_{\mathrm{R}}\times m_{\mathrm{A}}}$, $\mathbf{H}_{\mathrm{RB}}\in\mathbb{C}^{m_{\mathrm{B}}\times m_{\mathrm{R}}}$, and $\mathbf{H}_{\mathrm{RS}}\in\mathbb{C}^{m_{\mathrm{S}}\times m_{\mathrm{R}}}$, respectively. In this system, A communicates with B, while S seeks to estimate the transmitter--sensor channel $\mathbf{H}_{\mathrm{AS}}$, thereby enabling unauthorized sensing and inference about the transmitter.

\subsection{Transmit Signal Model}

The transmitter adopts a superposition-based signaling strategy that combines a message-bearing communication signal with a transmit-side AN component. Accordingly, the transmitted signal vector is expressed as
\begin{equation}
	\mathbf{x}=\mathbf{F}_{\mathrm{c}}\mathbf{w}_{\mathrm{c}}+\mathbf{F}_{\mathrm{s}}\mathbf{w}_{\mathrm{s}},
	\label{eq:tx_signal_from_A}
\end{equation}
where $\mathbf{w}_{\mathrm{c}}\in\mathbb{C}^{m_{\min}\times 1}$ denotes the communication symbol vector, with $1\leq m_{\min}\leq \min\{m_{\mathrm{A}},m_{\mathrm{B}}\}$, and $\mathbf{F}_{\mathrm{c}}\in\mathbb{C}^{m_{\mathrm{A}}\times m_{\min}}$ denotes the corresponding communication precoder. In addition, $\mathbf{w}_{\mathrm{s}}\in\mathbb{C}^{m_{\mathrm{A}}\times 1}$ denotes the transmit-side AN vector, and $\mathbf{F}_{\mathrm{s}}\in\mathbb{C}^{m_{\mathrm{A}}\times m_{\mathrm{A}}}$ is its associated precoder. Notably, the dimensionality of $\mathbf{F}_{\mathrm{s}}$ depends only on the transmitter-side dimensions and is therefore independent of the sensor antenna configuration. We assume that $\mathbb{E}\{\mathbf{w}_{\mathrm{c}}\}=\mathbb{E}\{\mathbf{w}_{\mathrm{s}}\}=\mathbf{0}$, $\mathbb{E}\{\mathbf{w}_{\mathrm{c}}\mathbf{w}_{\mathrm{c}}^{\mathsf{H}}\}=\mathbf{I}_{m_{\min}}$, $\mathbb{E}\{\mathbf{w}_{\mathrm{s}}\mathbf{w}_{\mathrm{s}}^{\mathsf{H}}\}=\mathbf{I}_{m_{\mathrm{A}}}$, and $\mathbb{E}\{\mathbf{w}_{\mathrm{c}}\mathbf{w}_{\mathrm{s}}^{\mathsf{H}}\}=\mathbf{0}$. Unless otherwise stated, we set $m_{\min}=\min\{m_{\mathrm{A}},m_{\mathrm{B}}\}$ for notational simplicity.


The AN component is introduced to impair the sensor's ability to estimate $\mathbf{H}_{\mathrm{AS}}$ accurately, thereby enhancing transmitter privacy. Since the AN signal is given by $\mathbf{F}_{\mathrm{s}}\mathbf{w}_{\mathrm{s}}$ and $\mathbb{E}\{\mathbf{w}_{\mathrm{s}}\mathbf{w}_{\mathrm{s}}^{\mathrm{H}}\}=\mathbf{I}$, its covariance is $\mathbf{F}_{\mathrm{s}}\mathbf{F}_{\mathrm{s}}^{\mathrm{H}}$, which shows that the AN spatial covariance is optimized implicitly through the precoder $\mathbf{F}_{\mathrm{s}}$. Moreover, the dedicated AN precoder provides an additional spatial degree of freedom for RIS-assisted channel manipulation. In particular, by jointly designing $\mathbf{F}_{\mathrm{s}}$ and the RIS reflection coefficients, the reflected interference can be spatially shaped and preferentially directed toward the sensor through the cascaded A--R--S link, while maintaining the desired communication performance at the legitimate receiver.
\begin{figure}
\centering 
\includegraphics[width=0.8\columnwidth]{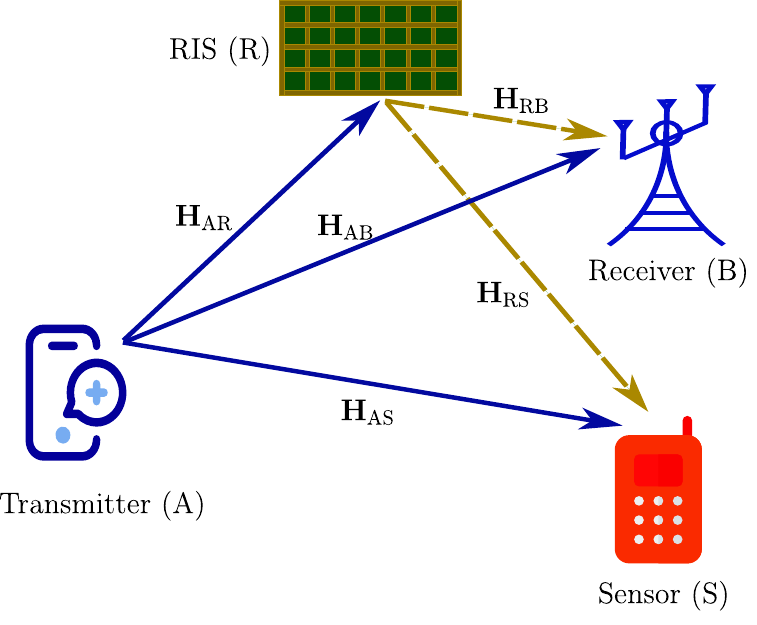}
\caption{A typical RIS-aided communication system with a malicious sensor.}
\label{fig:sysMod}
\end{figure}
\subsection{Communication Model}
Using~\eqref{eq:tx_signal_from_A}, the received signal at B is given by
\begin{equation}
	\mathbf{y}_{\mathrm{B}}=(\mathbf{H}_{\mathrm{AB}}+\mathbf{H}_{\mathrm{RB}}\boldsymbol{\Theta}\mathbf{H}_{\mathrm{AR}})\mathbf{x}+\mathbf{n}_{\mathrm{B}}=\mathbf{Z}_{\mathrm{AB}}\mathbf{x}+\mathbf{n}_{\mathrm{B}},
	\label{eq:rx_signal_at_B}
\end{equation}
where $\boldsymbol{\Theta}=\diag(\boldsymbol{\theta})$ denotes the RIS reflection matrix, and $\mathbf{n}_{\mathrm{B}}\sim\mathcal{CN}(\mathbf{0},\sigma_{\mathrm{B}}^{2}\mathbf{I}_{m_{\mathrm{B}}})$ denotes the additive white Gaussian noise (AWGN) vector at B. Moreover, the RIS reflection-coefficient vector is modeled as $\boldsymbol{\theta}=[\theta_{1},\theta_{2},\ldots,\theta_{m_{\mathrm{R}}}]\trans\in\mathbb{C}^{m_{\mathrm{R}}\times 1}$, where $\theta_{\varkappa}=\exp(j\phi_{\varkappa})$ with $\phi_{\varkappa}\in[0,2\pi)$ denoting the phase shift induced by the $\varkappa$-th RIS meta-atom, for all $\varkappa\in\mathcal{M}_{\mathrm{R}}\triangleq\{1,2,\ldots,m_{\mathrm{R}}\}$.

We assume that A has imperfect CSI for the A--B and R--B links.\footnote{This assumption is adopted for analytical tractability and to isolate the impact of RIS-assisted interference shaping and AN design. It is also standard in RIS-aided systems, since the A--R link can often be estimated more accurately through calibration and control signaling between the transmitter and the RIS controller, and its estimation errors are comparatively less critical to the communication-privacy trade-off considered in this work.} Specifically, we model $\mathbf{H}_{\mathrm{AB}}=\widehat{\mathbf{H}}_{\mathrm{AB}}+\bm{\Delta}_{\mathrm{AB}}$ and $\mathbf{H}_{\mathrm{RB}}=\widehat{\mathbf{H}}_{\mathrm{RB}}+\bm{\Delta}_{\mathrm{RB}}$, where $\widehat{\mathbf{H}}_{\mathrm{AB}}$ and $\widehat{\mathbf{H}}_{\mathrm{RB}}$ are the available channel estimates, and $\bm{\Delta}_{\mathrm{AB}}$ and $\bm{\Delta}_{\mathrm{RB}}$ denote the corresponding CSI error matrices. The entries of $\bm{\Delta}_{\mathrm{AB}}$ and $\bm{\Delta}_{\mathrm{RB}}$ are assumed to be independent and identically distributed (i.i.d.) and follow $\mathcal{CN}(0,\varsigma_{\mathrm{AB}}^{2})$ and $\mathcal{CN}(0,\varsigma_{\mathrm{RB}}^{2})$, respectively~\cite{24-TWC-KeshavSingh}.

Based on~\eqref{eq:rx_signal_at_B}, the achievable rate at B for a given $(\mathbf{F},\boldsymbol{\theta})$, measured in nats/s/Hz, is conservatively modeled as
\begin{equation}
	C_{\mathrm{AB}}(\mathbf{F},\boldsymbol{\theta})=\ln\det\big(\mathbf{I}+\widehat{\mathbf{Z}}_{\mathrm{AB}}\mathbf{F}_{\mathrm{c}}\mathbf{F}_{\mathrm{c}}\herm\widehat{\mathbf{Z}}_{\mathrm{AB}}\herm\mathbf{Q}_{\mathrm{AB}}^{-1}\big),\label{eq:C_AB}
\end{equation}
where $\mathbf{F}\triangleq[\mathbf{F}_{\mathrm{c}},\mathbf{F}_{\mathrm{s}}]$, $\widehat{\mathbf{Z}}_{\mathrm{AB}}=\widehat{\mathbf{H}}_{\mathrm{AB}}+\widehat{\mathbf{H}}_{\mathrm{RB}}\boldsymbol{\Theta}\mathbf{H}_{\mathrm{AR}}$ denotes the estimated composite A--B channel, and $\mathbf{Q}_{\mathrm{AB}}$ denotes the interference-plus-noise covariance matrix. Using the independence of $\bm{\Delta}_{\mathrm{AB}}$ and $\bm{\Delta}_{\mathrm{RB}}$, a conservative approximation\footnote{The approximation is conservative because the CSI-error terms are not exploited as part of the useful signal; instead, their contribution is entirely incorporated into $\mathbf{Q}_{\mathrm{AB}}$ as effective interference-plus-noise via their second-order moments. As a result, the rate expression is a pessimistic surrogate for the actual achievable rate under imperfect CSI.} of $\mathbf{Q}_{\mathrm{AB}}$ is given by
\begin{align}
	\mathbf{Q}_{\mathrm{AB}}=\  & \sigma_{\mathrm{B}}^{2}\mathbf{I}_{m_{\mathrm{B}}}+\widehat{\mathbf{Z}}_{\mathrm{AB}}\mathbf{F}_{\mathrm{s}}\mathbf{F}_{\mathrm{s}}\herm\widehat{\mathbf{Z}}_{\mathrm{AB}}\herm\nonumber \\
	& +\!\mathbb{E}\big\{\boldsymbol{\Delta}_{\mathrm{AB}}\big(\mathbf{F}_{\mathrm{c}}\mathbf{F}_{\mathrm{c}}\herm\!+\!\mathbf{F}_{\mathrm{s}}\mathbf{F}_{\mathrm{s}}\herm\big)\boldsymbol{\Delta}_{\mathrm{AB}}\herm\big\}\!\!+\!\!\mathbb{E}\big\{\big(\boldsymbol{\Delta}_{\mathrm{RB}}\boldsymbol{\Theta}\mathbf{H}_{\mathrm{AR}}\big)\nonumber \\
	& \times\big(\mathbf{F}_{\mathrm{c}}\mathbf{F}_{\mathrm{c}}\herm\!+\!\mathbf{F}_{\mathrm{s}}\mathbf{F}_{\mathrm{s}}\herm\big)\big(\boldsymbol{\Delta}_{\mathrm{RB}}\boldsymbol{\Theta}\mathbf{H}_{\mathrm{AR}}\big)\herm\big\}\nonumber \\
	\overset{\langle\texttt{a}\rangle}{=}\  & \widehat{\mathbf{Z}}_{\mathrm{AB}}\mathbf{F}_{\mathrm{s}}\mathbf{F}_{\mathrm{s}}\herm\widehat{\mathbf{Z}}_{\mathrm{AB}}\herm+\big[\sigma_{\mathrm{B}}^{2}+\varsigma_{\mathrm{AB}}^{2}\tr\big(\mathbf{F}_{\mathrm{c}}\mathbf{F}_{\mathrm{c}}\herm\!+\!\mathbf{F}_{\mathrm{s}}\mathbf{F}_{\mathrm{s}}\herm\big)\nonumber \\
	& +\varsigma_{\mathrm{RB}}^{2}\tr(\mathbf{H}_{\mathrm{AR}}\big(\mathbf{F}_{\mathrm{c}}\mathbf{F}_{\mathrm{c}}\herm\!+\!\mathbf{F}_{\mathrm{s}}\mathbf{F}_{\mathrm{s}}\herm\big)\mathbf{H}_{\mathrm{AR}}\herm)\big]\mathbf{I}_{m_{\mathrm{B}}}, \label{eq:Q_AB_ClosedForm}
\end{align}
where step $\langle\texttt{a}\rangle$ follows from the identity $\mathbb{E}\{\mathbf{X}\mathbf{Y}\mathbf{X}\herm\}=\sigma^{2}\tr(\mathbf{Y})\mathbf{I}$ for a random matrix $\mathbf{X}$ with i.i.d. entries distributed as $\mathcal{CN}(0,\sigma^{2})$ and any deterministic matrix $\mathbf{Y}$ of compatible dimension.
\subsection{Sensing Model}
By listening to the transmission from A, the sensor S attempts to estimate the channel $\mathbf{H}_{\mathrm{AS}}$. To this end, S collects the transmitted signal from A over $K$ time-slots. For analytical tractability, we assume that all channels remain unchanged over these $K$ time-slots. With a slight abuse of notation, we denote the transmit signal from A over $K$ time-slots by $\mathbf{X}=\mathbf{X}_{\mathrm{c}}+\mathbf{X}_{\mathrm{s}}$, where $\mathbf{X}_{\mathrm{c}}=\mathbf{F}_{\mathrm{c}}\mathbf{W}_{\mathrm{c}}$, $\mathbf{X}_{\mathrm{s}}=\mathbf{F}_{\mathrm{s}}\mathbf{W}_{\mathrm{s}}$, $\mathbf{W}_{\mathrm{c}}=[\mathbf{w}_{\mathrm{c},1},\ldots,\mathbf{w}_{\mathrm{c},K}]\in\mathbb{C}^{m_{\min}\times K}$, and $\mathbf{W}_{\mathrm{s}}=[\mathbf{w}_{\mathrm{s},1},\ldots,\mathbf{w}_{\mathrm{s},K}]\in\mathbb{C}^{m_{\mathrm{A}}\times K}$. Here, $\mathbf{w}_{\mathrm{c},k}$ and $\mathbf{w}_{\mathrm{s},k}$ denote the communication and AN vectors transmitted during the $k$-th time-slot, respectively, for $k\in\mathcal{K}\triangleq\{1,2,\ldots,K\}$. Accordingly, the observation at S can be written as
\begin{align}
	\mathbf{Y}_{\mathrm{S}}=\  & \big(\mathbf{H}_{\mathrm{AS}}+\mathbf{H}_{\mathrm{RS}}\boldsymbol{\Theta}\mathbf{H}_{\mathrm{AR}}\big)\mathbf{X}+\mathbf{N}_{\mathrm{S}}\nonumber \\
	=\  & \mathbf{H}_{\mathrm{AS}}\mathbf{X}+\mathbf{H}_{\mathrm{RS}}\boldsymbol{\Theta}\mathbf{H}_{\mathrm{AR}}\mathbf{X}+\mathbf{N}_{\mathrm{S}},\label{eq:observation_at_S}
\end{align}
where the entries of $\mathbf{N}_{\mathrm{S}}$ are i.i.d. and follow $\mathcal{CN}(0,\sigma_{\mathrm{S}}^{2})$, and $\mathbf{N}_{\mathrm{S}}$ denotes the AWGN matrix at S. By vectorizing~\eqref{eq:observation_at_S}, the observation can be expressed as
\begin{equation}
	\mathbf{y}_{\mathrm{S}}=\vecop(\mathbf{Y}_{\mathrm{S}})=\widetilde{\mathbf{X}}\mathbf{h}_{\mathrm{AS}}+\breve{\mathbf{X}}\mathbf{h}_{\mathrm{RS}}+\mathbf{n}_{\mathrm{S}},\label{eq:vec_observation_at_S}
\end{equation}
where $\widetilde{\mathbf{X}}=\mathbf{X}\trans\otimes\mathbf{I}_{m_{\mathrm{S}}}$, $\breve{\mathbf{X}}=\big(\mathbf{X}\trans\mathbf{H}_{\mathrm{AR}}\trans\bm{\Theta}\big)\otimes\mathbf{I}_{m_{\mathrm{S}}}$, $\mathbf{h}_{\mathrm{AS}}=\vecop(\mathbf{H}_{\mathrm{AS}})$, $\mathbf{h}_{\mathrm{RS}}=\vecop(\mathbf{H}_{\mathrm{RS}})$, and $\mathbf{n}_{\mathrm{S}}=\vecop(\mathbf{N}_{\mathrm{S}})$. Here, we have used the identity $\vecop(\mathbf{X}\mathbf{Y})=\vecop(\mathbf{I}\mathbf{X}\mathbf{Y})=\big(\mathbf{Y}\trans\otimes\mathbf{I}\big)\vecop(\mathbf{X})$. We model the true channel statistics of $\mathbf{h}_{\mathrm{AS}}$ and $\mathbf{h}_{\mathrm{RS}}$ as $\mathbf{h}_{\mathrm{AS}}\sim\mathcal{CN}(\bm{\mu}_{\mathrm{AS}},\bm{\Sigma}_{\mathrm{AS}})$ and $\mathbf{h}_{\mathrm{RS}}\sim\mathcal{CN}(\bm{0},\bm{\Sigma}_{\mathrm{RS}})$, respectively.\footnote{The non-zero-mean model for $\mathbf{h}_{\mathrm{AS}}$ is retained because $\mathbf{h}_{\mathrm{AS}}$ is the channel to be inferred by the malicious sensor, and allowing a non-zero mean provides a more general prior model that can capture a deterministic component. By contrast, the zero-mean assumption for $\mathbf{h}_{\mathrm{RS}}$ is adopted primarily for analytical tractability. Specifically, since $\mathbf{h}_{\mathrm{RS}}$ appears as a nuisance term in the cascaded RIS contribution, modeling it as zero-mean ensures that its effect is fully captured by $\bm{\Sigma}_{\mathrm{RS}}$, which considerably simplifies the derivation of the mismatched linear minimum mean square error (LMMSE) estimator and the associated Bayesian MSE expressions.} To account for imperfect prior knowledge at S, we assume that only erroneous statistical information is available, namely $\widehat{\bm{\mu}}_{\mathrm{S,AS}}$, $\widehat{\bm{\Sigma}}_{\mathrm{S,AS}}$, and $\widehat{\bm{\Sigma}}_{\mathrm{S,RS}}$, which denote the presumed mean of $\mathbf{h}_{\mathrm{AS}}$, the presumed covariance of $\mathbf{h}_{\mathrm{AS}}$, and the presumed covariance of $\mathbf{h}_{\mathrm{RS}}$, respectively.

Under the LMMSE framework~\cite{91_Scharf_Stat_Sig_Proc}, the estimate of $\mathbf{h}_{\mathrm{AS}}$ at S is given by
\begin{equation}
	\widehat{\mathbf{h}}_{\mathrm{S,AS}}=\widehat{\bm{\mu}}_{\mathrm{S,AS}}+\mathbf{R}_{\mathrm{S}}(\mathbf{y}_{\mathrm{S}}-\widetilde{\mathbf{X}}\widehat{\bm{\mu}}_{\mathrm{S,AS}}),\label{eq:estimate_at_S}
\end{equation}
where $\mathbf{R}_{\mathrm{S}}$ denotes the mismatched LMMSE gain and is given by
\begin{equation}
	\mathbf{R}_{\mathrm{S}}=\widehat{\bm{\Sigma}}_{\mathrm{S,AS}}\widetilde{\mathbf{X}}\herm\big(\widetilde{\mathbf{X}}\widehat{\bm{\Sigma}}_{\mathrm{S,AS}}\widetilde{\mathbf{X}}\herm\!+\!\breve{\mathbf{X}}_{\mathrm{S}}\widehat{\bm{\Sigma}}_{\mathrm{S,RS}}\breve{\mathbf{X}}_{\mathrm{S}}\herm\!+\!\sigma_{\mathrm{S}}^{2}\mathbf{I}_{m_{\mathrm{S}}m_{\mathrm{A}}}\big)^{-1},\label{eq:LMMSE_filter_at_S}
\end{equation}
where $\breve{\mathbf{X}}_{\mathrm{S}}=(\mathbf{X}\trans\widehat{\mathbf{H}}_{\mathrm{S,AR}}\trans\bm{\Theta})\otimes\mathbf{I}_{m_{\mathrm{S}}}$, and $\widehat{\mathbf{H}}_{\mathrm{S,AR}}=\mathbf{H}_{\mathrm{AR}}+\bm{\delta}_{\mathrm{S,AR}}$ denotes the estimate of $\mathbf{H}_{\mathrm{AR}}$ available at S. Note that $\mathbf{y}_{\mathrm{S}}$ is generated using the true cascaded term, and hence depends on $\breve{\mathbf{X}}$, whereas the sensor designs the LMMSE gain using its presumed model, and hence depends on $\breve{\mathbf{X}}_{\mathrm{S}}$ and the erroneous priors.

The following proposition characterizes the mean, covariance, and Bayesian MSE of the estimation error at S under mismatched prior information.

\begin{prop} \label{Prop-1}
	Define the estimation error as $\mathbf{e}_{\mathrm{S}}\triangleq\mathbf{h}_{\mathrm{AS}}-\widehat{\mathbf{h}}_{\mathrm{S,AS}}$. Then, for a given $(\mathbf{F},\bm{\theta})$, its mean is given by
	\begin{equation}
		\mathbb{E}\{\mathbf{e}_{\mathrm{S}}\}=\big(\mathbf{I}_{m_{\mathrm{S}}m_{\mathrm{A}}}\!-\!\mathbf{R}_{\mathrm{S}}\widetilde{\mathbf{X}}\big)\big(\bm{\mu}_{\mathrm{AS}}\!-\!\widehat{\bm{\mu}}_{\mathrm{S,AS}}\big),\label{eq:mean_eS}
	\end{equation}
	and its covariance is given by
	\begin{multline}
		\cov\{\mathbf{e}_{\mathrm{S}}\} = (\mathbf{I}_{m_{\mathrm{S}}m_{\mathrm{A}}}\!-\!\mathbf{R}_{\mathrm{S}}\widetilde{\mathbf{X}})\bm{\Sigma}_{\mathrm{AS}}(\mathbf{I}_{m_{\mathrm{S}}m_{\mathrm{A}}}\!-\!\mathbf{R}_{\mathrm{S}}\widetilde{\mathbf{X}})\herm \\ +\mathbf{R}_{\mathrm{S}}(\bm{\Sigma}_{\mathrm{ARS}}\!+\!\sigma_{\mathrm{S}}^{2}\mathbf{I}_{m_{\mathrm{S}}m_{\mathrm{A}}})\mathbf{R}_{\mathrm{S}}\herm,\label{eq:cov_S}
	\end{multline}
	where $\bm{\Sigma}_{\mathrm{ARS}}\triangleq\breve{\mathbf{X}}\bm{\Sigma}_{\mathrm{RS}}\breve{\mathbf{X}}\herm$. Consequently, the true Bayesian MSE at S is given by~\eqref{eq:true_MSE}, shown on the next page. 
\end{prop}

\begin{IEEEproof}
	See Appendix~\ref{sec:proof_Prop-1}.
\end{IEEEproof}

It is worth emphasizing that, in~\eqref{eq:true_MSE}, the filter $\mathbf{R}_{\mathrm{S}}$ is designed using the priors available at S, i.e., $\widehat{\bm{\Sigma}}_{\mathrm{S,AS}}$, $\widehat{\bm{\Sigma}}_{\mathrm{S,RS}}$, and $\breve{\mathbf{X}}_{\mathrm{S}}$, as indicated in~\eqref{eq:LMMSE_filter_at_S}.

Now, we assume that A also has erroneous prior information about $\mathbf{h}_{\mathrm{AS}}$ and $\mathbf{h}_{\mathrm{RS}}$, which is modeled as $\widehat{\mathbf{h}}_{\mathrm{A,AS}}\sim\mathcal{CN}(\widehat{\bm{\mu}}_{\mathrm{A,AS}},\widehat{\bm{\Sigma}}_{\mathrm{A,AS}})$ and $\widehat{\mathbf{h}}_{\mathrm{A,RS}}\sim\mathcal{CN}(\bm{0},\widehat{\bm{\Sigma}}_{\mathrm{A,RS}})$, respectively. Based on these priors, the corresponding mismatched MSE expression for $\mathbf{H}_{\mathrm{AS}}$ at A is given by~\eqref{eq:pred_MSE}, shown on the next page, 
\begin{figure*}
	\begin{multline}
		\mathsf{MSE}_{\mathrm{true}}(\mathbf{F},\bm{\theta}) = \|(\mathbf{I}_{m_{\mathrm{S}}m_{\mathrm{A}}}-\mathbf{R}_{\mathrm{S}}\widetilde{\mathbf{X}})(\bm{\mu}_{\mathrm{AS}}-\widehat{\bm{\mu}}_{\mathrm{S,AS}})\|^{2} \\
		+\tr\big[(\mathbf{I}_{m_{\mathrm{S}}m_{\mathrm{A}}}-\mathbf{R}_{\mathrm{S}}\widetilde{\mathbf{X}}) \bm{\Sigma}_{\mathrm{AS}}(\mathbf{I}_{m_{\mathrm{S}}m_{\mathrm{A}}}-\mathbf{R}_{\mathrm{S}}\widetilde{\mathbf{X}})\herm\big]+\tr\big[\mathbf{R}_{\mathrm{S}}(\bm{\Sigma}_{\mathrm{ARS}}+\sigma_{\mathrm{S}}^{2}\mathbf{I}_{m_{\mathrm{S}}m_{\mathrm{A}}})\mathbf{R}_{\mathrm{S}}\herm\big].\label{eq:true_MSE}
	\end{multline}
	\begin{multline}
		\mathsf{MSE}_{\mathrm{pred}}(\mathbf{F},\bm{\theta})
		= \|(\mathbf{I}_{m_{\mathrm{S}}m_{\mathrm{A}}}-\mathbf{R}_{\mathrm{A}}\widetilde{\mathbf{X}})(\bm{\mu}_{\mathrm{AS}}-\widehat{\bm{\mu}}_{\mathrm{A,AS}})\|^{2}\\
		+\tr\big[(\mathbf{I}_{m_{\mathrm{S}}m_{\mathrm{A}}}-\mathbf{R}_{\mathrm{A}}\widetilde{\mathbf{X}})\widehat{\bm{\Sigma}}_{\mathrm{A,AS}} (\mathbf{I}_{m_{\mathrm{S}}m_{\mathrm{A}}}-\mathbf{R}_{\mathrm{A}}\widetilde{\mathbf{X}})\herm\big] +\tr\big[\mathbf{R}_{\mathrm{A}}(\widehat{\bm{\Sigma}}_{\mathrm{A,ARS}}+\sigma_{\mathrm{S}}^{2}\mathbf{I}_{m_{\mathrm{S}}m_{\mathrm{A}}})\mathbf{R}_{\mathrm{A}}\herm\big].\label{eq:pred_MSE}
	\end{multline}
\hrulefill
\end{figure*}
where $\widehat{\bm{\Sigma}}_{\mathrm{A,ARS}}\triangleq\breve{\mathbf{X}}\widehat{\bm{\Sigma}}_{\mathrm{A,RS}}\breve{\mathbf{X}}\herm$, and the predicted LMMSE filter is defined as
\begin{equation}
	\mathbf{R}_{\mathrm{A}}=\widehat{\bm{\Sigma}}_{\mathrm{A,AS}}\widetilde{\mathbf{X}}\herm(\widetilde{\mathbf{X}}\widehat{\bm{\Sigma}}_{\mathrm{A,AS}}\widetilde{\mathbf{X}}\herm\!+\!\breve{\mathbf{X}}\widehat{\bm{\Sigma}}_{\mathrm{A,RS}}\breve{\mathbf{X}}\herm\!+\!\sigma_{\mathrm{S}}^{2}\mathbf{I}_{m_{\mathrm{S}}m_{\mathrm{A}}})^{-1}.\label{eq:LMMSE_filter_at_A}
\end{equation}
However, since A is unaware that its priors are erroneous, the bias term (i.e., $\|(\mathbf{I}_{m_{\mathrm{S}}m_{\mathrm{A}}}-\mathbf{R}_{\mathrm{A}}\widetilde{\mathbf{X}})(\bm{\mu}_{\mathrm{AS}}-\widehat{\bm{\mu}}_{\mathrm{A,AS}})\|^{2}$) is not accessible to A. Therefore, the predicted MSE available at A is given by
\begin{align}
	& \xi_{\mathrm{pred}}(\mathbf{F},\bm{\theta}) \notag \\
	= \ & \tr\big\{(\mathbf{I}_{m_{\mathrm{S}}m_{\mathrm{A}}}-\mathbf{R}_{\mathrm{A}}\widetilde{\mathbf{X}})\widehat{\bm{\Sigma}}_{\mathrm{A,AS}} (\mathbf{I}_{m_{\mathrm{S}}m_{\mathrm{A}}}-\mathbf{R}_{\mathrm{A}}\widetilde{\mathbf{X}})\herm\big\} \notag \\
	& \qquad \qquad \quad +\tr\big\{\mathbf{R}_{\mathrm{A}}(\widehat{\bm{\Sigma}}_{\mathrm{A,ARS}}+\sigma_{\mathrm{S}}^{2}\mathbf{I}_{m_{\mathrm{S}}m_{\mathrm{A}}})\mathbf{R}_{\mathrm{A}}\herm\big\}.\label{eq:unbiased_predicted_MSE}
\end{align}

\subsection{Problem Formulation for Transmitter Privacy}

We now formulate an optimization problem to jointly design the precoding matrix $\mathbf{F}$ and the passive beamforming vector $\bm{\theta}$ so as to maximize $\xi_{\mathrm{pred}}(\mathbf{F},\bm{\theta})$, while guaranteeing a communication QoS between A and B. The resulting optimization problem is given by
\begin{subequations}
	\label{eq:P0}
	\begin{align}
		\underset{\mathbf{F},\bm{\theta}}{\maximize}\  & \xi_{\mathrm{pred}}(\mathbf{F},\bm{\theta}),\label{eq:P0_obj}\\
		\st\  & C_{\mathrm{AB}}(\mathbf{F},\bm{\theta})\geq\mathscr{C}_{\mathrm{AB}},\label{eq:P0_rate_constraint}\\
		& \|\mathbf{F}\|_{\mathsf{F}}^{2}\leq p_{\mathrm{max}},\label{eq:P0_TPC}\\
		& |\theta_{\varkappa}|=1,\ \forall\varkappa\in\mathcal{M}_{\mathrm{R}}.\label{eq:P0_UMC}
	\end{align}
\end{subequations}
In~\eqref{eq:P0}, the objective is to maximize the predicted MSE at A. Constraint~\eqref{eq:P0_rate_constraint} ensures that the achievable rate at B is no smaller than the threshold $\mathscr{C}_{\mathrm{AB}}$, \eqref{eq:P0_TPC} imposes the maximum transmit-power budget $p_{\mathrm{max}}$, and \eqref{eq:P0_UMC} enforces the unit-modulus constraint on each RIS meta-atom. The formulation in~\eqref{eq:P0} captures the trade-off between reliable communication at B and transmitter privacy against S.

\begin{rem}
	It is worth emphasizing that, in~\eqref{eq:P0}, the objective is to maximize the predicted MSE at S based on the prior information available at A. Solving~\eqref{eq:P0} yields the \emph{optimal} transmit design $(\mathbf{F}_{\mathrm{opt}},\boldsymbol{\theta}_{\mathrm{opt}})$ to be employed by A. For performance evaluation, however, the \emph{true} MMSE achieved at S is computed using~\eqref{eq:true_MSE} for $(\mathbf{F}_{\mathrm{opt}},\boldsymbol{\theta}_{\mathrm{opt}})$.
	
	It is also important to note that the value of $\xi_{\mathrm{pred}}(\mathbf{F},\bm{\theta})$ can be very small, which may lead to numerical instability during the optimization process. To address this issue, we define the normalized objective
	\begin{equation}
		\bar{\xi}_{\mathrm{pred}}(\mathbf{F},\bm{\theta})=\frac{1}{\tr\big(\bm{\Sigma}_{\mathrm{AS}}\big)}\xi_{\mathrm{pred}}(\mathbf{F},\bm{\theta}).\label{eq:normalized_MSE_at_A}
	\end{equation}
	Using~\eqref{eq:normalized_MSE_at_A}, the optimization problem in~\eqref{eq:P0} can be equivalently reformulated as
	\begin{equation}
		\underset{\mathbf{F},\boldsymbol{\theta}}{\maximize}\big\{\bar{\xi}_{\mathrm{pred}}(\mathbf{F},\bm{\theta})\mid\eqref{eq:P0_rate_constraint},\eqref{eq:P0_TPC},\eqref{eq:P0_UMC}\big\}.\label{eq:P1}
	\end{equation}
	It is readily seen that~\eqref{eq:P1} is non-convex due to the objective function in~\eqref{eq:normalized_MSE_at_A}, the QoS constraint in~\eqref{eq:P0_rate_constraint}, and the unit-modulus constraint in~\eqref{eq:P0_UMC}. Moreover, the coupling between the optimization variables $\mathbf{F}$ and $\boldsymbol{\theta}$, together with the equality constraint in~\eqref{eq:P0_UMC}, makes the problem particularly challenging to solve. In addition, developing a low-complexity solution to~\eqref{eq:P1} is crucial, since the computational burden can become prohibitive when the number of RIS meta-atoms is large. Motivated by these challenges, the next section develops a numerically efficient solution for~\eqref{eq:P1}.
\end{rem}

\section{AO-Based Proposed Solution\label{sec:AO-Based-Proposed-Solution}}

Due to the non-convexity of the rate constraint in~\eqref{eq:P0_rate_constraint}, directly handling~\eqref{eq:P1} is challenging. To facilitate the subsequent algorithm design, we first reformulate this constraint. Specifically, for a non-negative real variable $\tau$, \eqref{eq:P0_rate_constraint} can be equivalently written as
\begin{equation}
	\underbrace{1+\tau-\frac{1}{\mathscr{C}_{\mathrm{AB}}}C_{\mathrm{AB}}(\mathbf{F},\boldsymbol{\theta})}_{\triangleq f(\mathbf{F},\boldsymbol{\theta},\mathscr{C}_{\mathrm{AB}},\tau)}=0.\label{eq:f_definition}
\end{equation}
Then, following the primal-dual decomposition (PDD) framework in~\cite{20-TSP-PDD}, the augmented Lagrangian associated with~\eqref{eq:P1} can be expressed as
\begin{multline}
	g_{\nu,\rho}(\mathbf{F},\boldsymbol{\theta},\tau)=\bar{\xi}_{\mathrm{pred}}(\mathbf{F},\bm{\theta})-\nu f(\mathbf{F},\boldsymbol{\theta},\mathscr{C}_{\mathrm{AB}},\tau)\\
	-\frac{1}{2\rho}f^{2}(\mathbf{F},\boldsymbol{\theta},\mathscr{C}_{\mathrm{AB}},\tau),\label{eq:augmented_objective}
\end{multline}
where $\nu$ is the Lagrange multiplier and $\rho$ is the penalty parameter. Using~\eqref{eq:augmented_objective}, the problem in~\eqref{eq:P1} can be transformed into
\begin{equation}
	\underset{\mathbf{F},\boldsymbol{\theta},\tau}{\maximize}\ \big\{g_{\nu,\rho}(\mathbf{F},\boldsymbol{\theta},\tau)\mid\tau\in\mathbb{R}_{+},\eqref{eq:P0_TPC},\eqref{eq:P0_UMC}\big\}.\label{eq:P2}
\end{equation}
It can be observed that the constraints in~\eqref{eq:P2} are now decoupled, whereas the design variables remain coupled in the objective function. Motivated by this structure, we adopt an AO-based scheme to obtain a stationary solution. In particular, we employ an alternating projected gradient ascent method to develop a low-complexity solution to~\eqref{eq:P2}. Before presenting the proposed algorithm, we first derive the complex-valued gradients of $g_{\nu,\rho}(\mathbf{F},\boldsymbol{\theta},\tau)$ with respect to $\mathbf{F}_{\mathrm{c}}$, $\mathbf{F}_{\mathrm{s}}$, and $\boldsymbol{\theta}$ in the following theorems. \setcounter{thm}{0}

\begin{thm}
	\label{thm:grad_Fc}
	A closed-form expression for $\nabla_{\mathbf{F}}g_{\nu,\rho}(\mathbf{F},\boldsymbol{\theta},\tau)$ is given by $\big[\nabla_{\mathbf{F}_{\mathrm{c}}}g_{\nu,\rho}(\mathbf{F},\boldsymbol{\theta},\tau),\nabla_{\mathbf{F}_{\mathrm{s}}}g_{\nu,\rho}(\mathbf{F},\boldsymbol{\theta},\tau)\big]$, where $\nabla_{\mathbf{F}_{\mathrm{c}}}g_{\nu,\rho}(\mathbf{F},\boldsymbol{\theta},\tau)$ and $\nabla_{\mathbf{F}_{\mathrm{s}}}g_{\nu,\rho}(\mathbf{F},\boldsymbol{\theta},\tau)$ are respectively given by~\eqref{eq:grad_Fc_Closed} and~\eqref{eq:grad_Fs_Closed}, shown on the next page.
\end{thm}
\begin{IEEEproof}
	See Appendix~\ref{sec:proof_grad_Fc}.
\end{IEEEproof}

\begin{figure*}[tbh]
	\begin{multline}
		\nabla_{\mathbf{F}_{\mathrm{c}}}g_{\nu,\rho}(\mathbf{F},\boldsymbol{\theta},\tau)=\frac{1}{\tr\big(\bm{\Sigma}_{\mathrm{AS}}\big)}\unvec_{m_{\mathrm{A}}\times m_{\min}}\big(\big[\vecop\trans\big(\mathbf{M}_{\mathrm{c}}\trans+\mathbf{U}_{\mathrm{c}}\trans\big)\big(\mathbf{I}_{m_{\mathrm{A}}}\otimes\mathbf{G}\big)\big]\trans\big)\\
		-\frac{1}{\mathscr{C}_{\mathrm{AB}}}\big\{\nu+\tfrac{1}{\rho}f(\mathbf{F},\boldsymbol{\theta},\mathscr{C}_{\mathrm{AB}},\tau)\big\}\big[\tr\big(\mathbf{D}_{\mathrm{AB}}\big)\big(\varsigma_{\mathrm{AB}}^{2}\mathbf{I}_{m_{\mathrm{A}}}+\varsigma_{\mathrm{RB}}^{2}\mathbf{H}_{\mathrm{AR}}\herm\mathbf{H}_{\mathrm{AR}}\big)-\widehat{\mathbf{Z}}_{\mathrm{AB}}\herm\mathbf{E}_{\mathrm{AB}}^{-1}\widehat{\mathbf{Z}}_{\mathrm{AB}}\big]\mathbf{F}_{\mathrm{c}}.\label{eq:grad_Fc_Closed}
	\end{multline}
	\begin{multline}
		\nabla_{\mathbf{F}_{\mathrm{s}}}g_{\nu,\rho}(\mathbf{F},\boldsymbol{\theta},\tau)=\frac{1}{\tr\big(\bm{\Sigma}_{\mathrm{AS}}\big)}\unvec_{m_{\mathrm{A}}\times m_{\mathrm{A}}}\big(\big[\vecop\trans\big(\mathbf{M}_{\mathrm{s}}\trans+\mathbf{U}_{\mathrm{s}}\trans\big)\big(\mathbf{I}_{m_{\mathrm{A}}}\otimes\mathbf{G}_{\mathrm{s}}\big)\big]\trans\big)\\
		-\frac{1}{\mathscr{C}_{\mathrm{AB}}}\big\{\nu+\tfrac{1}{\rho}f(\mathbf{F},\boldsymbol{\theta},\mathscr{C}_{\mathrm{AB}},\tau)\big\}\tr\big(\mathbf{D}_{\mathrm{AB}}\big)\big(\widehat{\mathbf{Z}}_{\mathrm{AB}}\herm\widehat{\mathbf{Z}}_{\mathrm{AB}}+\varsigma_{\mathrm{AB}}^{2}\mathbf{I}_{m_{\mathrm{A}}}+\varsigma_{\mathrm{RB}}^{2}\mathbf{H}_{\mathrm{AR}}\herm\mathbf{H}_{\mathrm{AR}}\big)\mathbf{F}_{\mathrm{s}}.\label{eq:grad_Fs_Closed}
	\end{multline}
	\rule{1\textwidth}{1pt}
\end{figure*}

\begin{thm}
	\label{thm:grad_theta}
	A closed-form expression for $\nabla_{\bm{\theta}}g_{\nu,\rho}(\mathbf{F},\boldsymbol{\theta},\tau)$ is given by~\eqref{eq:grad_theta_Closed}, shown on the next page.
\end{thm}
\begin{IEEEproof}
	See Appendix~\ref{sec:proof_grad_theta}.
\end{IEEEproof}

\begin{figure*}[tbh]
	\begin{multline}
		\nabla_{\bm{\theta}}g_{\nu,\rho}(\mathbf{F},\boldsymbol{\theta},\tau)=\vecd\Big[\frac{1}{\tr\big(\bm{\Sigma}_{\mathrm{AS}}\big)}\unvec_{m_{\mathrm{R}}\times m_{\mathrm{R}}}\big(\big[\vecop\trans\big(\mathbf{V}_{3}\trans+\mathbf{J}_{4}\trans\big)\big(\mathbf{I}_{m_{\mathrm{R}}}\otimes\mathbf{G}_{\theta}\big)\big]\trans\big)\Big]\\
		-\frac{1}{\mathscr{C}_{\mathrm{AB}}}\big\{\nu+\tfrac{1}{\rho}f(\mathbf{F},\boldsymbol{\theta},\mathscr{C}_{\mathrm{AB}},\tau)\big\}\vecd\big(\widehat{\mathbf{H}}_{\mathrm{RB}}\herm\mathbf{L}_{\mathrm{AB}}\mathbf{H}_{\mathrm{AR}}\herm\big).\label{eq:grad_theta_Closed}
	\end{multline}
	\rule{1\textwidth}{1pt}
\end{figure*}

\begin{algorithm}[t]
	\caption{AO-Based Proposed Algorithm to Solve~\eqref{eq:P2}}
	\label{algorithm-1}
	\KwIn{$\mathbf{F}^{(0)}$, $\boldsymbol{\theta}^{(0)}$, $\tau^{(0)}$, $\mu_{\mathbf{F}}$, $\mu_{\boldsymbol{\theta}}$, $\nu$, $\rho$, $\kappa$}
	\KwOut{$\mathbf{F}_{\mathsf{opt}}$, $\boldsymbol{\theta}_{\mathsf{opt}}$}
	$\jmath\leftarrow1$\;
	\Repeat{convergence}{
		\tcc{\textcolor{blue}{Update $(\mathbf{F},\bm{\theta},\tau)$ for fixed $(\nu,\rho)$}}
		Obtain $\big(\mathbf{F}^{(\jmath+1)},\bm{\theta}^{(\jmath+1)},\tau^{(\jmath+1)}\big)$ via \textbf{Algorithm~\ref{algorithm-2}}\;
		\tcc{\textcolor{blue}{Update $\nu$}}
		$\nu\leftarrow\nu+\frac{1}{\rho}f\big(\mathbf{F}^{(\jmath+1)},\bm{\theta}^{(\jmath+1)},\mathscr{C}_{\mathrm{AB}},\tau^{(\jmath+1)}\big)$\;
		\tcc{\textcolor{blue}{Update $\rho$}}
		$\rho\leftarrow\kappa\rho$\;
		\tcc{\textcolor{blue}{Update iteration counter}}
		$\jmath\leftarrow\jmath+1$\;
	}
	\tcc{\textcolor{blue}{Final assignment}}
	$\mathbf{F}_{\mathsf{opt}}\leftarrow\mathbf{F}^{(\jmath)}$, $\bm{\theta}_{\mathsf{opt}}\leftarrow\bm{\theta}^{(\jmath)}$\;
\end{algorithm}

\begin{algorithm}[tbh]
	\caption{Penalty-Based AO Algorithm to Solve~\eqref{eq:P2} for Fixed $(\nu,\rho)$}
	\label{algorithm-2}
	\KwIn{$\mathbf{F}^{(\jmath)}$, $\bm{\theta}^{(\jmath)}$, $\tau^{(\jmath)}$, $\mu_{\mathbf{F}}$, $\mu_{\bm{\theta}}$, $\nu$, $\rho$}
	\KwOut{$\mathbf{F}^{(\bar{\jmath}+1)}$, $\boldsymbol{\theta}^{(\bar{\jmath}+1)}$, $\tau^{(\bar{\jmath}+1)}$}
	\tcc{\textcolor{blue}{Initial assignment}}
	$\bar{\jmath}\leftarrow1$\;
	$\mathbf{F}^{(\bar{\jmath})}\leftarrow\mathbf{F}^{(\jmath)}$, $\bm{\theta}^{(\bar{\jmath})}\leftarrow\bm{\theta}^{(\jmath)}$, $\tau^{(\bar{\jmath})}\leftarrow\tau^{(\jmath)}$\;
	\Repeat{convergence}{
		\tcc{\textcolor{blue}{Update $\mathbf{F}$ for fixed $(\bm{\theta},\tau)$}}
		$\mathbf{F}^{(\bar{\jmath}+1)}\leftarrow\Pi_{\mathcal{F}}\big\{\mathbf{F}^{(\bar{\jmath})}+\mu_{\mathbf{F}}\nabla_{\mathbf{F}}\big(g_{\nu,\rho}\big(\mathbf{F}^{(\bar{\jmath})},\bm{\theta}^{(\bar{\jmath})},\tau^{(\bar{\jmath})}\big)\big)\big\}$\;
		\tcc{\textcolor{blue}{Update $\bm{\theta}$ for fixed $(\mathbf{F},\tau)$}}
		$\bm{\theta}^{(\bar{\jmath}+1)}\leftarrow\Pi_{\vartheta}\big\{\bm{\theta}^{(\bar{\jmath})}+\mu_{\bm{\theta}}\nabla_{\bm{\theta}}\big(g_{\nu,\rho}\big(\mathbf{F}^{(\bar{\jmath}+1)},\bm{\theta}^{(\bar{\jmath})},\tau^{(\bar{\jmath})}\big)\big)\big\}$\;
		\tcc{\textcolor{blue}{Update $\tau$ for fixed $(\mathbf{F},\bm{\theta})$}}
		$\tau^{(\bar{\jmath}+1)}\leftarrow\max\big\{0,\frac{1}{\mathscr{C}_{\mathrm{AB}}}C_{\mathrm{AB}}\big(\mathbf{F}^{(\bar{\jmath}+1)},\bm{\theta}^{(\bar{\jmath}+1)}\big)-1-\nu\rho\big\}$\;
		\tcc{\textcolor{blue}{Update iteration counter}}
		$\bar{\jmath}\leftarrow\bar{\jmath}+1$\;
	}
	\tcc{\textcolor{blue}{Update assignment}}
	$\mathbf{F}^{(\jmath+1)}\leftarrow\mathbf{F}^{(\bar{\jmath})}$, $\bm{\theta}^{(\jmath+1)}\leftarrow\bm{\theta}^{(\bar{\jmath})}$, $\tau^{(\jmath+1)}\leftarrow\tau^{(\bar{\jmath})}$\;
\end{algorithm}

Equipped with \textbf{Theorems~\ref{thm:grad_Fc}} and~\textbf{\ref{thm:grad_theta}}, we now develop an iterative AO framework to obtain a stationary solution to~\eqref{eq:P2}. The overall AO-based procedure is summarized in \textbf{Algorithm~\ref{algorithm-1}}. Specifically, the algorithm starts from initial values $\mathbf{F}^{(0)}$ and $\bm{\theta}^{(0)}$, while setting $\tau^{(0)}=\nu=0$, $\mu_{\mathbf{F}}=\mu_{\bm{\theta}}=100$, $\rho=10$, and $\kappa=0.1$. For given $(\nu,\rho)$, the variables $(\mathbf{F},\bm{\theta},\tau)$ are updated via \textbf{Algorithm~\ref{algorithm-2}}, whereas the Lagrange multiplier $\nu$ and the penalty parameter $\rho$ are updated in steps~4 and~5 of \textbf{Algorithm~\ref{algorithm-1}}, respectively.

In \textbf{Algorithm~\ref{algorithm-2}}, the transmit precoding matrix $\mathbf{F}$ is first updated for fixed $(\bm{\theta},\tau)$ in step~4, where $\mu_{\mathbf{F}}$ denotes the step-size and $\mathcal{F}$ denotes the feasible set of transmit precoders, defined as
\begin{equation}
	\mathcal{F}\triangleq\big\{\mathbf{F}\in\mathbb{C}^{m_{\mathrm{A}}\times(m_{\min}+m_{\mathrm{A}})}:\|\mathbf{F}\|_{\mathsf{F}}^{2}\leq p_{\max}\big\}.\label{eq:F_FeasibleSet}
\end{equation}
The projection onto the feasible set $\mathcal{F}$ is given by
\begin{equation}
	\Pi_{\mathcal{F}}\{\widetilde{\mathbf{F}}\}=\sqrt{p_{\max}}\ \widetilde{\mathbf{F}}/\max\big\{\|\widetilde{\mathbf{F}}\|_{\mathsf{F}},\sqrt{p_{\max}}\big\}.\label{eq:projection_onto_F}
\end{equation}
After updating $\mathbf{F}$, the RIS beamforming vector is updated. The corresponding feasible set $\vartheta$ is given by
\begin{equation}
	\vartheta\triangleq\big\{\bm{\theta}\in\mathbb{C}^{m_{\mathrm{R}}\times1}:|\theta_{\varkappa}|=1,\forall\varkappa\in\mathcal{M}_{\mathrm{R}}\big\},\label{eq:theta_FeasibleSet}
\end{equation}
and the projection onto $\vartheta$ is given by $\Pi_{\vartheta}\{\widetilde{\bm{\theta}}\}=[\bar{\theta}_{1},\ldots,\bar{\theta}_{m_{\mathrm{R}}}]\trans$, where for each $\varkappa\in\mathcal{M}_{\mathrm{R}}$, we have
\begin{equation}
	\bar{\theta}_{\varkappa}=\begin{cases}
		\widetilde{\theta}_{\varkappa}/|\widetilde{\theta}_{\varkappa}|, & \text{if }\widetilde{\theta}_{\varkappa}\neq0,\\
		\exp(j\phi),\ \phi\in[0,2\pi), & \text{otherwise}.
	\end{cases}\label{eq:projection_onto_theta}
\end{equation}
Finally, it is worth noting that although we initialize the step-sizes as $\mu_{\mathbf{F}}=\mu_{\bm{\theta}}=100$, their effective values at each iteration are determined via a \emph{backtracking line search} based on the Armijo--Goldstein condition~\cite{Armijo}.

\paragraph{Convergence Analysis}

The convergence of the proposed algorithm can be justified within the standard penalty dual decomposition framework combined with cyclic block coordinate optimization. For any fixed multiplier–penalty pair $(\nu,\rho)$, \textbf{Algorithm~\ref{algorithm-2}} generates a sequence by successively updating the blocks $\mathbf{F}$, $\bm{\theta}$, and $\tau$ through the corresponding subproblems of the augmented formulation in~\eqref{eq:P2}, while keeping the remaining blocks fixed. Since each block update is designed to optimize the augmented objective with respect to the corresponding variable block, the resulting inner-iteration objective sequence is monotonically nondecreasing. Furthermore, the transmit power constraint and the unit-modulus RIS constraint ensure that the feasible set associated with $(\mathbf{F},\bm{\theta})$ is compact, whereas the auxiliary variable $\tau$ is bounded by construction. Hence, the augmented objective is bounded from above over the feasible set, which guarantees convergence of the inner objective sequence; accordingly, the inner loop converges to a first-order stationary point of~\eqref{eq:P2} for the given $(\nu,\rho)$. In the outer loop, the multiplier update is given by $\nu^{(\jmath+1)}=\nu^{(\jmath)}+\frac{1}{\rho^{(\jmath)}}f(\mathbf{F}^{(\jmath+1)},\bm{\theta}^{(\jmath+1)},\tau^{(\jmath+1)})$, where $f(\mathbf{F},\bm{\theta},\tau)=0$ denotes the equality representation of the transformed QoS constraint. Equivalently, this yields $f(\mathbf{F}^{(\jmath+1)},\bm{\theta}^{(\jmath+1)},\tau^{(\jmath+1)})=\rho^{(\jmath)}\big(\nu^{(\jmath+1)}-\nu^{(\jmath)}\big)$. Therefore, if the multiplier sequence $\{\nu^{(\jmath)}\}$ is bounded and the penalty parameter satisfies $\rho^{(\jmath)}\to0$, then the equality-constraint residual vanishes asymptotically, namely, $f(\mathbf{F}^{(\jmath+1)},\bm{\theta}^{(\jmath+1)},\tau^{(\jmath+1)})\to0$, thereby establishing asymptotic primal feasibility. Consequently, every accumulation point generated by \textbf{Algorithm~\ref{algorithm-1}} is feasible for the transformed problem~\eqref{eq:P2} and satisfies its first-order optimality condition, and therefore fulfills the Karush–Kuhn–Tucker (KKT) conditions of~\eqref{eq:P2}. Finally, since~\eqref{eq:f_definition} constitutes an equivalent reformulation of the QoS constraint in~\eqref{eq:P0_rate_constraint}, the transformed problem in~\eqref{eq:P2} is equivalent to the original problem in~\eqref{eq:P0}. It follows that any accumulation point satisfying the KKT conditions of~\eqref{eq:P2} is also a KKT point of~\eqref{eq:P0}.

\paragraph{Complexity Analysis}

It is evident that the computational complexity of the proposed PDD-AO framework is dominated by the inner updates in \textbf{Algorithm~\ref{algorithm-2}}. Hence, the complexity of \textbf{Algorithm~\ref{algorithm-1}} can be quantified by counting the required number of complex multiplications in one inner iteration of \textbf{Algorithm~\ref{algorithm-2}} and then multiplying by the numbers of outer and inner iterations. Let $m_{\mathbf{F}}\triangleq m_{\min}+m_{\mathrm{A}}$ denote the number of columns of $\mathbf{F}=[\mathbf{F}_{\mathrm{c}},\mathbf{F}_{\mathrm{s}}]$, and let $I_{\mathrm{out}}$ and $I_{\mathrm{in}}$ denote the numbers of outer and inner iterations, respectively. First, consider the update of $\mathbf{F}$ in Step~4 of \textbf{Algorithm~\ref{algorithm-2}}. From \textbf{Theorem}~\textbf{\ref{thm:grad_Fc}}, the dominant cost in evaluating $\nabla_{\mathbf{F}}g_{\nu,\rho}(\mathbf{F},\bm{\theta},\tau)$ comes from: i) forming the matrix $\mathbf{X}=\mathbf{F}\mathbf{W}$ and the associated cascaded sensing terms, which requires $\mathcal{O}(m_{A}n_{\mathbf{F}}K+m_{R}m_{A}Km_{S})$ operations; ii) constructing and inverting the predicted LMMSE covariance matrix $\widetilde{\mathbf{X}}\widehat{\bm{\Sigma}}_{\mathrm{A,AS}}\widetilde{\mathbf{X}}\herm+\breve{\mathbf{X}}\widehat{\bm{\Sigma}}_{\mathrm{A,RS}}\breve{\mathbf{X}}\herm+\sigma_{\mathrm{S}}^{2}\mathbf{I}$ of dimension $Km_{\mathrm{S}}\times Km_{\mathrm{S}}$, whose complexity is $\mathcal{O}((Km_{\mathrm{S}})^{3})$; and iii) evaluating the communication-related matrices $\mathbf{Q}_{\mathrm{AB}}$, $\mathbf{E}_{\mathrm{AB}}$, and $\mathbf{D}_{\mathrm{AB}}$, whose dominant cost is $\mathcal{O}(m_{\mathrm{B}}^{3})$, while the remaining matrix products contribute $\mathcal{O}(m_{\mathrm{A}}n_{\mathbf{F}}m_{\mathrm{B}})$. Moreover, the projection onto the Frobenius-norm ball in~\eqref{eq:projection_onto_F} has complexity $\mathcal{O}(m_{\mathrm{A}}n_{\mathbf{F}})$. Therefore, the total complexity associated with the $\mathbf{F}$-update is $\mathcal{O}((Km_{\mathrm{S}})^{3}+m_{\mathrm{B}}^{3}+m_{\mathrm{A}}n_{\mathbf{F}}(K+m_{\mathrm{B}})+m_{\mathrm{R}}m_{\mathrm{A}}Km_{\mathrm{S}})$. Next, consider the update of $\bm{\theta}$ in Step~5 of \textbf{Algorithm~\ref{algorithm-2}}. From \textbf{Theorem}~\textbf{\ref{thm:grad_theta}}, the dominant operations in evaluating $\nabla_{\bm{\theta}}g_{\nu,\rho}(\mathbf{F},\bm{\theta},\tau)$ again include the inversion of the same $Km_{\mathrm{S}}\times Km_{\mathrm{S}}$ covariance matrix and the computation of the communication-side inverse/factorization, which contribute $\mathcal{O}((Km_{\mathrm{S}})^{3})$ and $\mathcal{O}(m_{\mathrm{B}}^{3})$, respectively. However, unlike a naïve implementation, the gradient in~\eqref{eq:grad_theta_Closed} involves the operator $\vecd(\cdot)$, so only the diagonal entries of the intermediate $m_{\mathrm{R}}\times m_{\mathrm{R}}$ matrices are required. In addition, the commutation matrices appearing in the derivative expressions are used only as permutation operators and need not be explicitly constructed. Consequently, the RIS-related arithmetic is obtained by directly evaluating the required $m_{\mathrm{R}}$ diagonal bilinear forms, which incurs complexity $\mathcal{O}(m_{\mathrm{R}}m_{\mathrm{A}}(Km_{\mathrm{S}}+m_{\mathrm{B}}))$, while the projection onto the unit-modulus set in~\eqref{eq:projection_onto_theta} is elementwise and has negligible complexity. Thus, the total complexity of the $\bm{\theta}$-update is $\mathcal{O}((Km_{\mathrm{S}})^{3}+m_{\mathrm{B}}^{3}+m_{\mathrm{A}}n_{\mathbf{F}}K+m_{\mathrm{R}}m_{\mathrm{A}}(Km_{\mathrm{S}}+m_{\mathrm{B}}))$. Finally, the update of $\tau$ in Step~6 is available in closed form and therefore has negligible complexity. Putting these observations together, the per-inner-iteration complexity of \textbf{Algorithm~\ref{algorithm-2}} can be expressed as $\mathcal{O}\!\big(2(Km_{\mathrm{S}})^{3}+2m_{\mathrm{B}}^{3}+m_{\mathrm{A}}n_{\mathbf{F}}(2K+m_{\mathrm{B}})+m_{\mathrm{R}}m_{\mathrm{A}}(2Km_{\mathrm{S}}+m_{\mathrm{B}})\big).$

Accordingly, the overall computational complexity of \textbf{Algorithm~\ref{algorithm-1}} is given by
\begin{multline*}
\mathcal{O}\!\big(I_{\mathrm{out}}I_{\mathrm{in}}\big[2(Km_{\mathrm{S}})^{3}+2m_{\mathrm{B}}^{3}+m_{\mathrm{A}}n_{\mathbf{F}}(2K+m_{\mathrm{B}})\\
+m_{\mathrm{R}}m_{\mathrm{A}}(2Km_{\mathrm{S}}+m_{\mathrm{B}})\big]\big).
\end{multline*}
It is worth noting that, because only the diagonal terms are needed in the $\bm{\theta}$-gradient and the commutation matrices are implemented implicitly, the dependence on the number of RIS meta-atoms $m_{\mathrm{R}}$ is linear rather than quadratic or cubic. Hence, under a practical implementation, the dominant computational burden arises from the sensing- and communication-side matrix inversions, whereas the RIS-related operations remain scalable even for large $m_{\mathrm{R}}$. 

\section{Results and Discussion\label{sec:Results-and-Discussion}}

In this section, we provide comprehensive numerical results to evaluate the performance of the system under consideration and provide detailed discussion on the results to obtain important system design insights. For this purpose, we first provide the details of the system setup, followed by numerical results and corresponding discussion. 

\subsection{System Setup}

The transmitter (A), receiver (B), RIS (R), and sensor (S) are located at $(0,0,0)\,\mathrm{m}$, $(100,20,5)\,\mathrm{m}$, $(50,10,5)\,\mathrm{m}$, and $(20,5,0)\,\mathrm{m}$, respectively. The adopted node geometry is \emph{similar} to representative setups commonly used in the RIS literature, e.g.,~\cite{21_TCOM_GeneralDistance, 21_TCOM_ChannelEstimation_Distance}. The small-scale fading on the A–B, A–S, and R–B links is modeled as Rician with Rician factor $3\,\mathrm{dB}$, whereas the R–S link is modeled as Rayleigh faded. The A–R link is assumed to be purely line-of-sight. The large-scale path loss between any two nodes is modeled as $-30-10\alpha\log_{10}(d/d_{0})\,\mathrm{dB}$, where $d$ denotes the inter-node distance, $d_{0}=1\,\mathrm{m}$ is the reference distance, and $\alpha$ is the path-loss exponent; specifically, $\alpha=3.6$ for the A–B and A–S links, and $\alpha=2.2$ for the A–R, R–S, and R–B links. Unless stated otherwise, we set $m_{\mathrm{A}}=m_{\mathrm{S}}=4$, $m_{\mathrm{B}}=16$, $m_{\mathrm{R}}=64$, $K=16$, $m_{\min}=\min\{m_{\mathrm{A}},m_{\mathrm{B}}\}$, $p_{\max}=10\,\mathrm{dBm}$, and $\mathscr{C}_{\mathrm{AB}}=5\,\mathrm{nats/s/Hz}$. The carrier frequency and system bandwidth are $2\,\mathrm{GHz}$ and $20\,\mathrm{MHz}$, respectively, and the noise power spectral density is $-174\,\mathrm{dBm/Hz}$, which yields $\sigma_{\mathrm{B}}^{2}=\sigma_{\mathrm{S}}^{2}=\sigma^{2}=10^{\frac{-174-30}{10}}\times20\times10^{6}\,\mathrm{W}$. Furthermore, we set $\varsigma_{\mathrm{AB}}^{2}=\varsigma_{\mathrm{RB}}^{2}=10^{2}\sigma^{2}$. The spatial correlation matrices at A, B, and S follow Hermitian Toeplitz structures corresponding to uniform linear arrays (ULAs) with exponential correlation and adjacent-element correlation coefficient $0.5$, whereas the RIS correlation matrix is generated according to~\cite{RIS_Correlation_Emil} with inter-meta-atom spacing $\lambda/4$. To model prior mismatch, we set $\widehat{\bm{\Sigma}}_{\mathrm{A,AS}}=\bm{\Sigma}_{\mathrm{AS}}+\varsigma_{\mathrm{A,AS}}^{2}\mathbf{I}_{m_{\mathrm{S}}m_{\mathrm{A}}}$, $\widehat{\bm{\Sigma}}_{\mathrm{S,AS}}=\bm{\Sigma}_{\mathrm{AS}}+\varsigma_{\mathrm{S,AS}}^{2}\mathbf{I}_{m_{\mathrm{S}}m_{\mathrm{A}}}$, $\widehat{\bm{\Sigma}}_{\mathrm{A,RS}}=\bm{\Sigma}_{\mathrm{RS}}+\varsigma_{\mathrm{A,RS}}^{2}\mathbf{I}_{m_{\mathrm{S}}m_{\mathrm{R}}}$, and $\widehat{\bm{\Sigma}}_{\mathrm{S,RS}}=\bm{\Sigma}_{\mathrm{RS}}+\varsigma_{\mathrm{S,RS}}^{2}\mathbf{I}_{m_{\mathrm{S}}m_{\mathrm{R}}}$, together with $\widehat{\bm{\mu}}_{\mathrm{A,AS}}=\bm{\mu}_{\mathrm{AS}}+\varsigma_{\mathrm{A,AS}}\mathbf{r}_{\mathrm{A,AS}}$ and $\widehat{\bm{\mu}}_{\mathrm{S,AS}}=\bm{\mu}_{\mathrm{AS}}+\varsigma_{\mathrm{S,AS}}\mathbf{r}_{\mathrm{S,AS}}$, where $\mathbf{r}_{\mathrm{A,AS}},\mathbf{r}_{\mathrm{S,AS}}\sim\mathcal{CN}(\mathbf{0},\mathbf{I})$. Unless stated otherwise, the prior-error variances are chosen as $\varsigma_{\mathrm{A,AS}}^{2}=\varsigma_{\mathrm{S,AS}}^{2}=0$ and $\varsigma_{\mathrm{A,RS}}^{2}=\varsigma_{\mathrm{S,RS}}^{2}=0$ under perfect-prior assumptions, and as $\varsigma_{\mathrm{A,AS}}^{2}=\varsigma_{\mathrm{S,AS}}^{2}=(5\times10^{5})\sigma^{2}$ and $\varsigma_{\mathrm{A,RS}}^{2}=\varsigma_{\mathrm{S,RS}}^{2}=(5\times10^{5})\sigma^{2}$ under imperfect-prior assumptions. All NMSE results are averaged over $1000$ independent channel realizations.

\subsection{Convergence Results}

\begin{figure}
\centering 
\includegraphics[width=0.68\columnwidth]{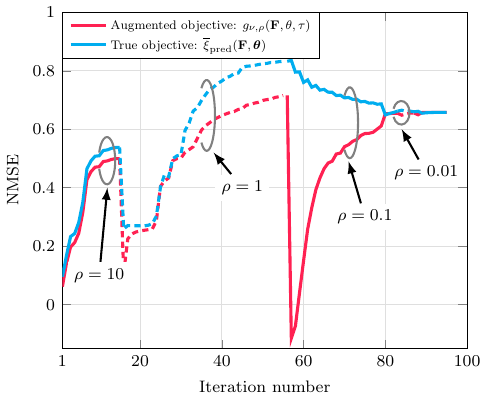}
\caption{Convergence behavior of the proposed AO algorithm.}
\label{fig:conv}
\end{figure}

In Fig.~\ref{fig:conv}, we illustrate the convergence behavior of the proposed AO framework in \textbf{Algorithm~\ref{algorithm-1}} for the RIS-aided setting under perfect prior knowledge of the $\mathbf{H}_{\mathrm{AS}}$ and $\mathbf{H}_{\mathrm{RS}}$ links at both A and S. The figure shows that the algorithm reaches convergence within a limited number of iterations. Starting from $\nu=0$ and $\rho=10$, the inner AO loop in ~\textbf{Algorithm~\ref{algorithm-2}} updates $(\mathbf{F},\bm{\theta},\tau)$ for fixed $(\nu,\rho)$, thereby yielding a monotonic increase of the augmented Lagrangian $g_{\nu,\rho}(\mathbf{F},\bm{\theta},\tau)$ until convergence of the inner subproblem. Subsequently, the multiplier $\nu$ and the penalty parameter $\rho$ are updated according to \textbf{Algorithm~\ref{algorithm-1}}, and the process is repeated. The sharp drops observed in the augmented objective at the outer iterations are expected, since the updates of $\nu$ and $\rho$ jointly tighten the enforcement of the equality constraint. More specifically, reducing $\rho$ strengthens the quadratic penalty component, while the multiplier update $\nu\leftarrow\nu+\frac{1}{\rho}f(\mathbf{F},\bm{\theta},C_{\mathrm{AB}},\tau)$ also increases the influence of the linear dual term associated with the constraint residual. Consequently, after each outer update, constraint violations are penalized more aggressively, which leads to the observed downward jump in the augmented objective. Moreover, the noticeable mismatch between the true objective and the augmented objective during the early iterations indicates that, although the iterates improve the penalized problem in~\eqref{eq:P2}, they are not yet primal-feasible for the original formulation in~\eqref{eq:P1}, i.e., the residual $f(\mathbf{F},\bm{\theta},C_{\mathrm{AB}},\tau)$ remains nonzero. As the outer loop advances, the joint updates of $\nu$ and $\rho$ drive this residual toward zero, so that the iterates eventually satisfy the original constraint as well, and the gap between the augmented and true objectives consequently vanishes.

\subsection{Impact of the Number of RIS Meta-Atoms}

\begin{figure}
\centering 
\includegraphics[width=0.68\columnwidth]{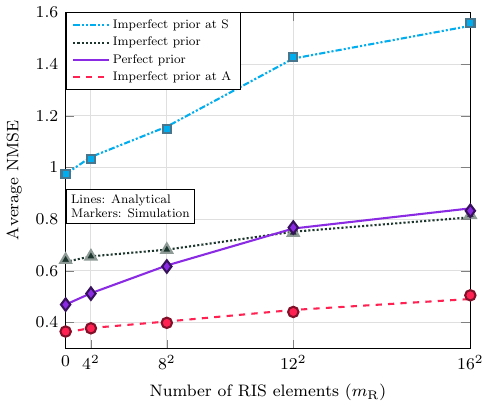}
\caption{Impact of the number of meta-atoms at the RIS on the NMSE at S.}
\label{fig:vary_mR}
\end{figure}

\begin{figure*}
\centering
\begin{minipage}{.32\textwidth}
  \centering
    \includegraphics[width=1\columnwidth]{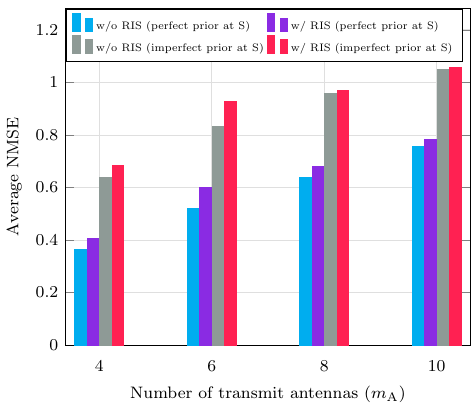}
\caption{Impact of the number of transmit antennas on the NMSE at S.}
\label{fig:vary_mA}
\end{minipage}%
\hfill 
\begin{minipage}{.32\textwidth}
  \centering
  \includegraphics[width=1\columnwidth]{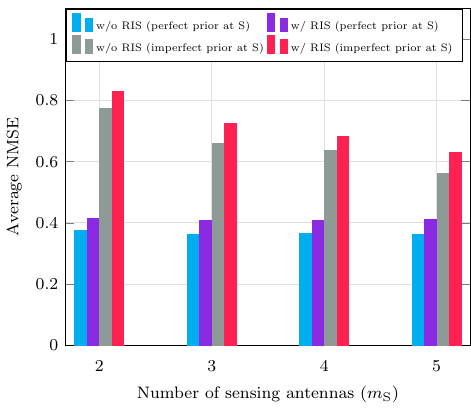}
\caption{Impact of the number of sensor antennas on the NMSE at S.}
\label{fig:vary_mS}
\end{minipage}%
\hfill 
\begin{minipage}{.32\textwidth}
  \centering
  \includegraphics[width=1\columnwidth]{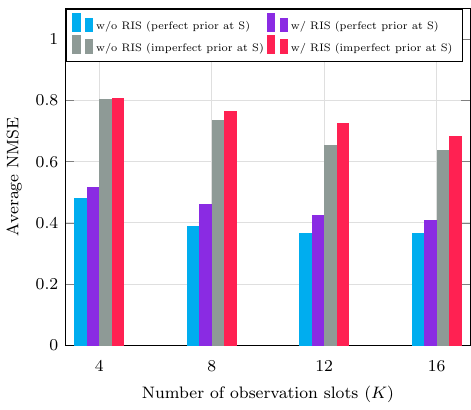}
\caption{Impact of the observation length on the NMSE at S.}
\label{fig:vary_K}
\end{minipage}%
\end{figure*}

In Fig.~\ref{fig:vary_mR}, we examine the impact of the number of RIS meta-atoms on the average NMSE at the sensor for four different prior-knowledge configurations. Here, ``Perfect prior'' denotes the case in which both A and S possess correct statistical priors, ``Imperfect prior'' corresponds to the case in which both nodes operate with erroneous priors, and ``Imperfect prior at A'' (resp. ``Imperfect prior at S'') denotes the case in which only A (resp. S) has imperfect prior knowledge while the other node has perfect priors. The solid curves are obtained analytically from the closed-form true Bayesian MSE expression in \eqref{eq:true_MSE}, evaluated at the optimized design $(\mathbf{F},\bm{\theta})$ and then normalized to produce the NMSE, whereas the markers are obtained by Monte Carlo simulation, where the received signal at S is generated according to the sensing model, the mismatched LMMSE estimator in \eqref{eq:estimate_at_S}--\eqref{eq:LMMSE_filter_at_S} is applied, and the resulting channel-estimation error is averaged over random realizations. It is observed that the analytical curves closely match the simulation markers for all values of $m_{\mathrm{R}}$, thereby validating the accuracy of the developed NMSE characterization. Moreover, the average NMSE increases monotonically with $m_{\mathrm{R}}$ in all four cases, which shows that a larger RIS offers additional spatial degrees of freedom for shaping the reflected signal over the A--R--S path and, consequently, for degrading the sensing capability at S. In particular, the point $m_{\mathrm{R}}=0$ corresponds to the non-RIS scenario, and it yields the smallest NMSE in each prior setting; therefore, the non-RIS case is the least favorable from the transmitter-privacy perspective. This observation clearly demonstrates the privacy advantage enabled by RIS deployment. The figure also reveals a clear ordering among the four prior settings: the highest NMSE is achieved when the prior is imperfect at S, followed by the case in which both A and S have imperfect priors, then the perfect-prior case, whereas the lowest NMSE is obtained when the prior is imperfect only at A. This behavior is physically intuitive. When A has imperfect priors, it optimizes $(\mathbf{F},\bm{\theta})$ using a mismatched predicted objective, so the resulting design is no longer well aligned with the true channel statistics governing the sensor's estimator; as a result, the actual degradation induced at S is reduced, leading to a lower NMSE. In contrast, when S has imperfect priors, its Wiener filter is itself mismatched to the true observation model, which directly deteriorates channel-estimation accuracy and leads to a larger NMSE. When both nodes have imperfect priors, these two effects coexist: the prior mismatch at S still degrades the estimator, whereas the prior mismatch at A partially weakens the effectiveness of the privacy-oriented design. Consequently, the corresponding performance lies between the cases of ``Imperfect prior at S'' and ``Perfect prior''.
\subsection{Impact of the Number of Transmit Antennas}

For the remaining figures, we restrict attention to the case of imperfect prior at A and consider two representative sub-cases at S, namely, ``perfect prior at S'' and ``imperfect prior at S'', since these two settings most clearly reveal the privacy implications of transmitter-side prior mismatch under different sensing conditions at the adversarial sensor. In Fig.~\ref{fig:vary_mA}, we investigate the impact of the number of transmit antennas on the average NMSE at the sensor for both RIS-aided and non-RIS architectures. It is observed that the average NMSE increases monotonically with $m_{\mathrm{A}}$ in all four cases. This trend can be understood from two complementary perspectives. First, for fixed $K$ and $m_{\mathrm{S}}$, increasing $m_{\mathrm{A}}$ enlarges the dimension of the unknown channel $\mathbf{H}_{\mathrm{AS}}\in\mathbb{C}^{m_{\mathrm{S}}\times m_{\mathrm{A}}}$ to be estimated, whereas the size of the observation available at the sensor, $\mathbf{Y}_{\mathrm{S}}\in\mathbb{C}^{m_{\mathrm{S}}\times K}$, remains unchanged. Equivalently, in the vectorized model, the number of unknown coefficients in $\mathbf{h}_{\mathrm{AS}}$ grows with $m_{\mathrm{A}}$, while the number of observation dimensions remains fixed at $Km_{\mathrm{S}}$. Hence, the estimation problem at S becomes progressively more challenging as $m_{\mathrm{A}}$ increases, which naturally results in a higher NMSE. Second, a larger transmit array provides additional spatial degrees of freedom for the design of the transmit precoder, and, in the RIS-aided case, these degrees of freedom can be further exploited jointly with the RIS phase profile to more effectively impair sensing at S while maintaining the communication requirement at B. This explains why, for every value of $m_{\mathrm{A}}$, the RIS-aided scheme consistently outperforms its non-RIS counterpart under the same prior setting. Moreover, the ordering between the perfect-prior and imperfect-prior cases at S remains consistent with that already observed in Fig.~\ref{fig:vary_mR}. Overall, Fig.~\ref{fig:vary_mA} shows that increasing the number of transmit antennas is beneficial for transmitter privacy, and that RIS assistance further amplifies this gain.

\subsection{Impact of the Number of Sensor Antennas}

In Fig.~\ref{fig:vary_mS}, we investigate the impact of the number of sensor antennas on the average NMSE for both RIS-aided and non-RIS architectures, while keeping the transmitter-side prior imperfect and considering the two representative cases of perfect prior at S and imperfect prior at S. A first noteworthy observation is that, when S has perfect prior knowledge, the average NMSE remains nearly unchanged as $m_{\mathrm{S}}$ increases. This behavior can be understood by noting that increasing $m_{\mathrm{S}}$ enlarges not only the observation matrix at the sensor but also the dimension of the channel $\mathbf{H}_{\mathrm{AS}}\in\mathbb{C}^{m_{\mathrm{S}}\times m_{\mathrm{A}}}$ to be estimated. Hence, the sensor acquires a larger spatial observation space, but it must simultaneously infer a proportionally larger number of channel coefficients. Since the estimator at S is statistically matched in the perfect-prior case, these two effects largely balance each other in normalized terms, and the resulting average NMSE changes only marginally with $m_{\mathrm{S}}$. In contrast, when S operates with imperfect prior knowledge, the average NMSE decreases as $m_{\mathrm{S}}$ increases. In this case, the additional sensing antennas provide a richer spatial observation that partially compensates for the prior mismatch, thereby improving the conditioning of the estimation problem and reducing the mismatch-induced error. The figure further shows that, for both prior settings at S and for all values of $m_{\mathrm{S}}$, the RIS-aided architecture consistently achieves a higher NMSE than its non-RIS counterpart. Therefore, although increasing the sensing capability of the adversarial node can alleviate the impact of prior mismatch at S, RIS assistance continues to provide a clear privacy advantage throughout the entire range of $m_{\mathrm{S}}$.

\subsection{Impact of the Observation Length}

\begin{figure*}
\centering
\begin{minipage}{.32\textwidth}
  \includegraphics[width=1\columnwidth]{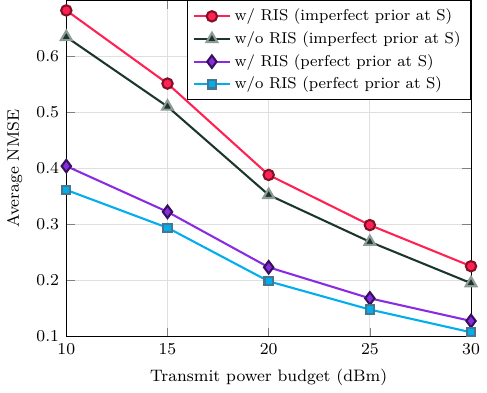}
\caption{Impact of transmit power budget on the NMSE at S.}
\label{fig:vary_p_max}
\end{minipage}%
\hfill 
\begin{minipage}{.32\textwidth}
  \centering \includegraphics[width=1\columnwidth]{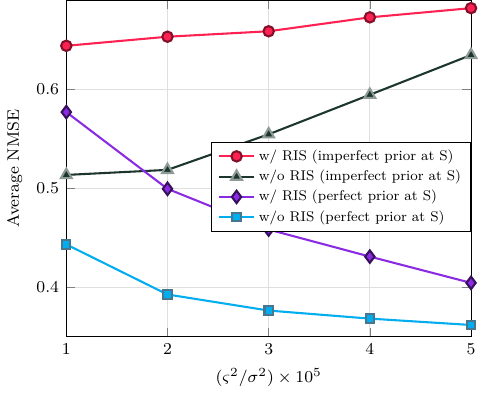}
\caption{Impact of imperfect prior on the NMSE at S.}
\label{fig:vary_prior}
\end{minipage}%
\hfill 
\begin{minipage}{.32\textwidth}
  \centering
  \includegraphics[width=1\columnwidth]{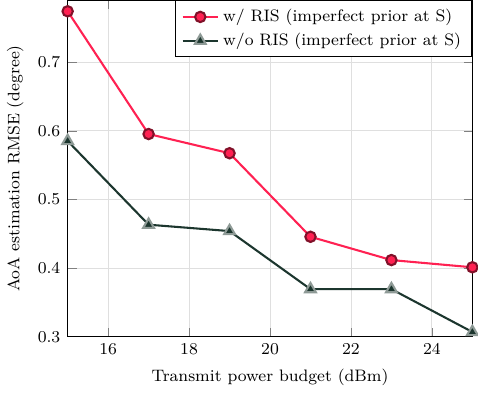}
	\caption{Impact of transmit power budget on the AoA estimation RMSE at S.}
	\label{fig:AoA_RMSE}
\end{minipage}%
\end{figure*}

In Fig.~\ref{fig:vary_K}, we study the effect of the observation length $K$ on the average NMSE at the sensor. As expected, the average NMSE decreases with $K$ for all considered cases. This behavior follows directly from the fact that a larger observation horizon provides the sensor with more temporal samples, thereby increasing the effective number of measurements available for channel inference and improving the accuracy of the LMMSE estimator. In other words, a longer sensing duration reduces the uncertainty in the channel estimate and strengthens the sensing capability of the adversarial node. Nevertheless, the RIS-aided scheme consistently yields a higher NMSE than the non-RIS scheme for both prior settings at S across the entire range of $K$. This is an important observation, since it shows that the privacy benefit of RIS assistance is not limited to short observation windows; rather, it persists even when the sensor is allowed to collect a substantially larger number of samples. Thus, while increasing $K$ is beneficial for the sensor from an estimation standpoint, the proposed RIS-aided design remains effective in preserving transmitter privacy.

\subsection{Impact of Transmit Power Budget}

In Fig.~\ref{fig:vary_p_max}, we examine the impact of the transmit power budget on the average NMSE at the sensor for both RIS-aided and non-RIS architectures, while restricting attention to the case of imperfect prior at A and considering the two representative subcases of perfect prior at S and imperfect prior at S. It is observed that the average NMSE decreases monotonically with $p_{\max}$ in all four cases. Although a larger power budget gives A greater flexibility to design $(\mathbf{F},\bm{\theta})$ in a privacy-aware manner, the dominant effect is that the received observation at S becomes stronger as $p_{\max}$ increases. More specifically, the same transmit signal $\mathbf{X}$ simultaneously drives both the desired sensing component $\mathbf{H}_{\mathrm{AS}}\mathbf{X}$ and the RIS-induced nuisance component $\mathbf{H}_{\mathrm{RS}}\bm{\Theta}\mathbf{H}_{\mathrm{AR}}\mathbf{X}$. Hence, increasing $p_{\max}$ does not merely intensify interference at S; it also strengthens the direct A–S observation from which S estimates the channel. Since the cascaded RIS term constitutes a structured impairment rather than an independently powered jammer, the net effect of increasing $p_{\max}$ is an improvement in the effective observation quality at S, which leads to a lower NMSE. Nevertheless, for every value of $p_{\max}$, the RIS-aided scheme consistently yields a higher NMSE than the corresponding non-RIS benchmark under the same prior setting. This shows that, although a larger transmit power budget improves the sensing capability of the adversarial node, RIS-assisted propagation shaping continues to preserve a non-negligible privacy advantage across the entire operating range.

\subsection{Impact of the Imperfect Priors}

In Fig.~\ref{fig:vary_prior}, we investigate the impact of the prior-error variance on the average NMSE at the sensor for both RIS-aided and non-RIS architectures, while restricting attention to the case of imperfect prior at~A and considering the two representative subcases of perfect prior at~S and imperfect prior at~S. Here, whenever prior mismatch is present, the corresponding prior-error variances are set equal to a common value $\varsigma^{2}$. A clear contrast is observed between the two sensing conditions at~S. When~S has perfect prior knowledge, the average NMSE decreases as $\varsigma^{2}$ increases. This behavior is physically intuitive, since a larger prior mismatch at~A makes the design of $(\mathbf{F},\bm{\theta})$ increasingly driven by an inaccurate predicted objective. As a result, the resulting privacy-oriented design becomes less aligned with the true channel statistics governing the sensor's estimator, so that the actual impairment induced at~S is weakened and the sensor can estimate the channel more accurately, thereby reducing the NMSE. In contrast, when S also operates with imperfect prior knowledge, the average NMSE increases monotonically with $\varsigma^{2}$. In this case, increasing $\varsigma^{2}$ not only preserves the transmitter-side design mismatch, but also directly worsens the mismatch in the Wiener filter employed at~S, and this estimator mismatch becomes the dominant effect, leading to a larger NMSE. The figure further shows that, for every value of $\varsigma^{2}$ and under both prior settings at~S, the RIS-aided architecture consistently achieves a higher NMSE than its non-RIS counterpart. Therefore, although severe transmitter-side prior mismatch can partially reduce the effectiveness of the privacy-oriented design, RIS assistance continues to provide a clear privacy advantage, and this advantage becomes particularly pronounced when the sensor itself also suffers from prior mismatch.

\subsection{AoA Privacy Interpretation}


To connect transmitter privacy to a more direct sensing metric, we next evaluate the AoA estimation performance at the malicious sensor using the estimated channel $\widehat{\mathbf{H}}_{\mathrm{AS}}$, via a Bartlett spatial-spectrum search over candidate angles. We restrict attention here to the case of imperfect prior at S, since when S has perfect prior knowledge, the prior mean may already contain deterministic transmitter-direction information, and therefore AoA estimation from $\widehat{\mathbf{H}}_{\mathrm{AS}}$ is no longer a meaningful privacy indicator. Fig.~\ref{fig:AoA_RMSE} shows that the root mean square error (RMSE) of AoA estimation decreases with the transmit-power budget for both the RIS-aided and non-RIS schemes, which is consistent with the trend already observed in Fig.~\ref{fig:vary_p_max} for the channel-estimation NMSE, i.e., as $p_{\max}$ increases, the observation quality at S improves, and the resulting estimate of $\mathbf{H}_{\mathrm{AS}}$ becomes more accurate. Nevertheless, the RIS-aided scheme consistently yields a higher AoA RMSE than the non-RIS benchmark over the entire range of $p_{\max}$. This observation is important because it shows that the privacy benefit of the proposed design is not limited to the intermediate channel-estimation metric, but also translates into degraded inference of a physically meaningful transmitter attribute, namely, its angle of arrival at the sensor. Therefore, RIS-assisted propagation shaping helps preserve transmitter privacy not only in terms of channel-estimation error, but also in terms of the AoA estimation accuracy achieved at the sensor.

\section{Conclusion\label{sec:Conclusion}}

This paper studied transmitter privacy in an RIS-aided multi-antenna wireless system in the presence of a malicious sensing node seeking to estimate the transmitter--sensor channel. A privacy-oriented joint optimization framework was developed for the transmit precoder and RIS reflection coefficients, with the aim of maximizing the malicious sensor's predicted channel-estimation error while guaranteeing the desired communication performance at the legitimate receiver under power and RIS constraints. To tackle the resulting non-convex problem, an AO-based method was proposed. The numerical results showed that RIS-assisted propagation shaping can significantly impair unauthorized channel estimation and provide clear privacy gains relative to non-RIS benchmarks, without sacrificing reliable communication. They further showed that these gains are also reflected in degraded AoA estimation at the malicious sensor, thereby providing a more direct sensing-level interpretation of transmitter privacy. 

An important direction for future work is to extend the proposed framework to multi-user scenarios with multiple legitimate receivers and multiple possibly colluding malicious sensors, where spatially distributed attackers may share observations and/or local channel estimates, leading to a richer joint estimation model and potentially different privacy--communication tradeoffs.

\appendices{}

\section{\label{sec:proof_Prop-1}Proof of Proposition~\ref{Prop-1}}

Substituting~\eqref{eq:vec_observation_at_S} into~\eqref{eq:estimate_at_S}, the estimation error can be written as
\begin{align}
	\mathbf{e}_{\mathrm{S}}=\  & \mathbf{h}_{\mathrm{AS}}-\widehat{\mathbf{h}}_{\mathrm{S,AS}}=\mathbf{h}_{\mathrm{AS}}-\big\{\widehat{\bm{\mu}}_{\mathrm{S,AS}}+\mathbf{R}_{\mathrm{S}}\big(\mathbf{y}_{\mathrm{S}}-\widetilde{\mathbf{X}}\widehat{\bm{\mu}}_{\mathrm{S,AS}}\big)\big\}\nonumber \\
	=\  & \mathbf{h}_{\mathrm{AS}}\!-\!\big\{\widehat{\bm{\mu}}_{\mathrm{S,AS}}\!+\!\mathbf{R}_{\mathrm{S}}\big(\widetilde{\mathbf{X}}\mathbf{h}_{\mathrm{AS}}\!+\!\breve{\mathbf{X}}\mathbf{h}_{\mathrm{RS}}\!+\!\mathbf{n}_{\mathrm{S}}\!-\!\widetilde{\mathbf{X}}\widehat{\bm{\mu}}_{\mathrm{S,AS}}\big)\big\}\nonumber \\
	=\  & \big(\mathbf{h}_{\mathrm{AS}}\!-\!\mathbf{R}_{\mathrm{S}}\widetilde{\mathbf{X}}\mathbf{h}_{\mathrm{AS}}\big)\!\nonumber \\
	& \qquad-\!\big(\widehat{\bm{\mu}}_{\mathrm{S,AS}}-\mathbf{R}_{\mathrm{S}}\widetilde{\mathbf{X}}\widehat{\bm{\mu}}_{\mathrm{S,AS}}\big)\!-\!\mathbf{R}_{\mathrm{S}}\big(\breve{\mathbf{X}}\mathbf{h}_{\mathrm{RS}}+\mathbf{n}_{\mathrm{S}}\big)\nonumber \\
	=\  & \big(\mathbf{I}_{m_{\mathrm{S}}m_{\mathrm{A}}}-\mathbf{R}_{\mathrm{S}}\widetilde{\mathbf{X}}\big)\big(\mathbf{h}_{\mathrm{AS}}-\widehat{\bm{\mu}}_{\mathrm{S,AS}}\big)-\mathbf{R}_{\mathrm{S}}\big(\breve{\mathbf{X}}\mathbf{h}_{\mathrm{RS}}+\mathbf{n}_{\mathrm{S}}\big)\nonumber \\
	=\  & \big(\mathbf{I}_{m_{\mathrm{S}}m_{\mathrm{A}}}\!-\!\mathbf{R}_{\mathrm{S}}\widetilde{\mathbf{X}}\big)\big(\mathbf{h}_{\mathrm{AS}}\!-\!\bm{\mu}_{\mathrm{AS}}\big)\!\nonumber \\
	& +\!\big(\mathbf{I}_{m_{\mathrm{S}}m_{\mathrm{A}}}\!-\!\mathbf{R}_{\mathrm{S}}\widetilde{\mathbf{X}}\big)\big(\bm{\mu}_{\mathrm{AS}}\!-\!\widehat{\bm{\mu}}_{\mathrm{S,AS}}\big)\!-\!\mathbf{R}_{\mathrm{S}}\big(\breve{\mathbf{X}}\mathbf{h}_{\mathrm{RS}}+\mathbf{n}_{\mathrm{S}}\big).\label{eq:appendix_eS}
\end{align}
Since $\mathbb{E}\{\mathbf{h}_{\mathrm{AS}}-\bm{\mu}_{\mathrm{AS}}\}=\mathbf{0}$, $\mathbb{E}\{\mathbf{h}_{\mathrm{RS}}\}=\mathbf{0}$, and $\mathbb{E}\{\mathbf{n}_{\mathrm{S}}\}=\mathbf{0}$, taking expectation on both sides of~\eqref{eq:appendix_eS} yields
\begin{equation}
	\mathbb{E}\{\mathbf{e}_{\mathrm{S}}\}=\big(\mathbf{I}_{m_{\mathrm{S}}m_{\mathrm{A}}}-\mathbf{R}_{\mathrm{S}}\widetilde{\mathbf{X}}\big)\big(\bm{\mu}_{\mathrm{AS}}-\widehat{\bm{\mu}}_{\mathrm{S,AS}}\big),
\end{equation}
which proves~\eqref{eq:mean_eS}.

Next, define the zero-mean error fluctuation as
\begin{align}
	\widetilde{\mathbf{e}}_{\mathrm{S}}=\  & \mathbf{e}_{\mathrm{S}}-\mathbb{E}\{\mathbf{e}_{\mathrm{S}}\}\nonumber \\
	=\  & \big(\mathbf{I}_{m_{\mathrm{S}}m_{\mathrm{A}}}-\mathbf{R}_{\mathrm{S}}\widetilde{\mathbf{X}}\big)\big(\mathbf{h}_{\mathrm{AS}}-\bm{\mu}_{\mathrm{AS}}\big)-\mathbf{R}_{\mathrm{S}}\big(\breve{\mathbf{X}}\mathbf{h}_{\mathrm{RS}}+\mathbf{n}_{\mathrm{S}}\big).\label{eq:appendix_tilde_eS}
\end{align}
Using~\eqref{eq:appendix_tilde_eS}, the error covariance is given by
\begin{align}
	& \cov\{\mathbf{e}_{\mathrm{S}}\}= \mathbb{E}\big\{\widetilde{\mathbf{e}}_{\mathrm{S}}\widetilde{\mathbf{e}}_{\mathrm{S}}\herm\big\}\nonumber \\
	=\  & \mathbb{E}\big\{\big(\mathbf{I}_{m_{\mathrm{S}}m_{\mathrm{A}}}-\mathbf{R}_{\mathrm{S}}\widetilde{\mathbf{X}}\big)\big(\mathbf{h}_{\mathrm{AS}}-\bm{\mu}_{\mathrm{AS}}\big) \nonumber \\
	& \qquad \qquad \times \big(\mathbf{h}_{\mathrm{AS}}-\bm{\mu}_{\mathrm{AS}}\big)\herm\big(\mathbf{I}_{m_{\mathrm{S}}m_{\mathrm{A}}}-\mathbf{R}_{\mathrm{S}}\widetilde{\mathbf{X}}\big)\herm\big\} \nonumber \\
	& +\mathbb{E}\big\{\mathbf{R}_{\mathrm{S}}\big(\breve{\mathbf{X}}\mathbf{h}_{\mathrm{RS}}+\mathbf{n}_{\mathrm{S}}\big)\big(\breve{\mathbf{X}}\mathbf{h}_{\mathrm{RS}}+\mathbf{n}_{\mathrm{S}}\big)\herm\mathbf{R}_{\mathrm{S}}\herm\big\}\nonumber \\
	& -\mathbb{E}\big\{\big(\mathbf{I}_{m_{\mathrm{S}}m_{\mathrm{A}}}\!-\!\mathbf{R}_{\mathrm{S}}\widetilde{\mathbf{X}}\big)\big(\mathbf{h}_{\mathrm{AS}}\!-\!\bm{\mu}_{\mathrm{AS}}\big)\big(\breve{\mathbf{X}}\mathbf{h}_{\mathrm{RS}}\!+\!\mathbf{n}_{\mathrm{S}}\big)\herm\mathbf{R}_{\mathrm{S}}\herm\big\}\nonumber \\
	& -\mathbb{E}\big\{\mathbf{R}_{\mathrm{S}}\big(\breve{\mathbf{X}}\mathbf{h}_{\mathrm{RS}}\!+\!\mathbf{n}_{\mathrm{S}}\big)\big(\mathbf{h}_{\mathrm{AS}}\!-\!\bm{\mu}_{\mathrm{AS}}\big)\herm\big(\mathbf{I}_{m_{\mathrm{S}}m_{\mathrm{A}}}\!-\!\mathbf{R}_{\mathrm{S}}\widetilde{\mathbf{X}}\big)\herm\big\}.\label{eq:appendix_cov_expand_1}
\end{align}
Because $\mathbf{h}_{\mathrm{AS}}-\bm{\mu}_{\mathrm{AS}}$, $\mathbf{h}_{\mathrm{RS}}$, and $\mathbf{n}_{\mathrm{S}}$ are mutually independent and zero-mean, the last two cross terms in~\eqref{eq:appendix_cov_expand_1} vanish. Moreover,
\begin{equation}
	\mathbb{E}\big\{(\mathbf{h}_{\mathrm{AS}}-\bm{\mu}_{\mathrm{AS}})(\mathbf{h}_{\mathrm{AS}}-\bm{\mu}_{\mathrm{AS}})\herm\big\}=\bm{\Sigma}_{\mathrm{AS}}, \label{eq:appendix_cov_expand_2}
\end{equation}
and
\begin{align}
	& \mathbb{E}\big\{\big(\breve{\mathbf{X}}\mathbf{h}_{\mathrm{RS}}+\mathbf{n}_{\mathrm{S}}\big)\big(\breve{\mathbf{X}}\mathbf{h}_{\mathrm{RS}}+\mathbf{n}_{\mathrm{S}}\big)\herm\big\} \nonumber \\
	= \ &  \breve{\mathbf{X}}\mathbb{E}\big\{\mathbf{h}_{\mathrm{RS}}\mathbf{h}_{\mathrm{RS}}\herm\big\}\breve{\mathbf{X}}\herm+\mathbb{E}\big\{\mathbf{n}_{\mathrm{S}}\mathbf{n}_{\mathrm{S}}\herm\big\}\nonumber \\
	= \  & \breve{\mathbf{X}}\bm{\Sigma}_{\mathrm{RS}}\breve{\mathbf{X}}\herm+\sigma_{\mathrm{S}}^{2}\mathbf{I}
	= \bm{\Sigma}_{\mathrm{ARS}}+\sigma_{\mathrm{S}}^{2}\mathbf{I}. \label{eq:appendix_cov_expand_3}
\end{align}
Using~\eqref{eq:appendix_cov_expand_2} and~\eqref{eq:appendix_cov_expand_3}, the expression in~\eqref{eq:appendix_cov_expand_1} becomes equal to~\eqref{eq:cov_S}. 

Finally, since vectorization preserves the Frobenius norm, we have
\begin{equation}
	\|\mathbf{H}_{\mathrm{AS}}-\widehat{\mathbf{H}}_{\mathrm{S,AS}}\|_{\mathsf{F}}^{2}=\|\mathbf{h}_{\mathrm{AS}}-\widehat{\mathbf{h}}_{\mathrm{S,AS}}\|^{2}=\|\mathbf{e}_{\mathrm{S}}\|^{2}. \label{eq:vec_preservation}
\end{equation}
Using the identity $\mathbb{E}\{\|\mathbf{x}\|^{2}\}=\|\mathbb{E}\{\mathbf{x}\}\|^{2}+\tr(\cov\{\mathbf{x}\})$ for any random vector $\mathbf{x}$, and~\eqref{eq:vec_preservation}, the true Bayesian MSE at S is given by~\eqref{eq:true_MSE}. This completes the proof.

\section{\label{sec:proof_grad_Fc}Proof of Theorem~\ref{thm:grad_Fc}}

With $\mathbf{F}_{\mathrm{s}}$, $\boldsymbol{\theta}$ and $\tau$ being fixed, the complex-valued differential of $g_{\nu,\rho}(\mathbf{F},\boldsymbol{\theta},\tau)$ w.r.t. $\mathbf{F}_{\mathrm{c}}$ can be obtained as 
\begin{multline}
\op{d}g_{\nu,\rho}(\mathbf{F},\boldsymbol{\theta},\tau)=\op{d}\bar{\xi}_{\mathrm{pred}}(\mathbf{F},\bm{\theta})\\
-\big\{\nu+\tfrac{1}{\rho}f(\mathbf{F},\boldsymbol{\theta},\mathscr{C}_{\mathrm{AB}},\tau)\big\}\op{d}f(\mathbf{F},\boldsymbol{\theta},\mathscr{C}_{\mathrm{AB}},\tau).\label{eq:A-1}
\end{multline}
For the first term on the RHS, we have 
\begin{align}
 & \op{d}\bar{\xi}_{\mathrm{pred}}(\mathbf{F},\bm{\theta})\nonumber \\
=\  & \op{d}\frac{1}{\tr\big(\bm{\Sigma}_{\mathrm{AS}}\big)}\Big[\tr\big\{(\mathbf{I}-\mathbf{R}_{\mathrm{A}}\widetilde{\mathbf{X}})\widehat{\bm{\Sigma}}_{\mathrm{A,AS}}(\mathbf{I}-\mathbf{R}_{\mathrm{A}}\widetilde{\mathbf{X}})\herm\big\}\nonumber \\
 & \qquad\qquad\qquad\qquad+\tr\big\{\mathbf{R}_{\mathrm{A}}(\widehat{\bm{\Sigma}}_{\mathrm{A,ARS}}+\sigma_{\mathrm{S}}^{2}\mathbf{I})\mathbf{R}_{\mathrm{A}}\herm\big\}\Big]\nonumber \\
=\  & \frac{1}{\tr\big(\bm{\Sigma}_{\mathrm{AS}}\big)}\Big[\op{d}\tr\big\{(\mathbf{I}-\mathbf{R}_{\mathrm{A}}\widetilde{\mathbf{X}})\widehat{\bm{\Sigma}}_{\mathrm{A,AS}}(\mathbf{I}-\mathbf{R}_{\mathrm{A}}\widetilde{\mathbf{X}})\herm\big\}\nonumber \\
 & \qquad\qquad\qquad+\op{d}\tr\big\{\mathbf{R}_{\mathrm{A}}(\widehat{\bm{\Sigma}}_{\mathrm{A,ARS}}+\sigma_{\mathrm{S}}^{2}\mathbf{I})\mathbf{R}_{\mathrm{A}}\herm\big\}\Big].\label{eq:A-2}
\end{align}
By representing $\mathbf{R}_{\mathrm{A}}=\mathbf{R}_{\mathrm{A}1}\mathbf{R}_{\mathrm{A}2}^{-1}$, the expression for $\op{d}\tr\big\{(\mathbf{I}-\mathbf{R}_{\mathrm{A}}\widetilde{\mathbf{X}})\widehat{\bm{\Sigma}}_{\mathrm{A,AS}}(\mathbf{I}-\mathbf{R}_{\mathrm{A}}\widetilde{\mathbf{X}})\herm\big\}$ can be obtained as 
\begin{align}
 & \op{d}\tr\big\{(\mathbf{I}-\mathbf{R}_{\mathrm{A}}\widetilde{\mathbf{X}})\widehat{\bm{\Sigma}}_{\mathrm{A,AS}}(\mathbf{I}-\mathbf{R}_{\mathrm{A}}\widetilde{\mathbf{X}})\herm\big\}\nonumber \\
=\  & \tr\big\{\op{d}(\mathbf{I}-\mathbf{R}_{\mathrm{A}}\widetilde{\mathbf{X}})\mathbf{M}_{1}\big\}+\tr\big\{\mathbf{M}_{1}\herm\op{d}(\mathbf{I}-\mathbf{R}_{\mathrm{A}}\widetilde{\mathbf{X}})\herm\big\}\nonumber \\
=\  & \tr\big\{\mathbf{M}_{2}\big(\op{d}\mathbf{R}_{\mathrm{A}2}\big)\big\}-\tr\big\{\mathbf{M}_{3}\big(\op{d}\widetilde{\mathbf{X}}\herm\big)\big\}\nonumber \\
=\  & \tr\big\{\big(\mathbf{M}_{4\mathrm{c}}+\mathbf{M}_{5\mathrm{c}}\big)\op{d}\big(\mathbf{F}_{\mathrm{c}}\conj\otimes\mathbf{I}_{m_{\mathrm{S}}}\big)\big\}\nonumber \\
=\  & \tr\big\{\mathbf{M}_{\mathrm{c}}\op{d}\big(\mathbf{F}_{\mathrm{c}}\conj\otimes\mathbf{I}_{m_{\mathrm{S}}}\big)\big\}\nonumber \\
=\  & \vecop\trans\big(\mathbf{M}_{\mathrm{c}}\trans\big)\op{d}\big\{\vecop\big(\mathbf{F}_{\mathrm{c}}\conj\otimes\mathbf{I}_{m_{\mathrm{S}}}\big)\big\}\nonumber \\
=\  & \vecop\trans\big(\mathbf{M}_{\mathrm{c}}\trans\big)\big(\mathbf{I}_{m_{\mathrm{A}}}\otimes\mathbf{G}_{\mathrm{c}}\big)\op{d}\big\{\vecop\big(\mathbf{F}_{\mathrm{c}}\conj\big)\big\},\label{eq:A-3}
\end{align}
where
\begin{subequations}
\label{eq:def_M_comm}
\begin{align}
 & \!\!\!\!\mathbf{M}_{1}=\widehat{\bm{\Sigma}}_{\mathrm{A,AS}}(\mathbf{I}-\mathbf{R}_{\mathrm{A}}\widetilde{\mathbf{X}})\herm,\label{eq:M1}\\
 & \!\!\!\!\mathbf{M}_{2}=\mathbf{R}_{\mathrm{A}2}^{-1}\widetilde{\mathbf{X}}\mathbf{M}_{1}\mathbf{R}_{\mathrm{A}}+\mathbf{R}_{\mathrm{A}}\herm\mathbf{M}_{1}\herm\widetilde{\mathbf{X}}\herm\mathbf{R}_{\mathrm{A}2}^{-1},\label{eq:M2}\\
 & \!\!\!\!\mathbf{M}_{3}=\mathbf{R}_{\mathrm{A}2}^{-1}\widetilde{\mathbf{X}}\mathbf{M}_{1}\widehat{\bm{\Sigma}}_{\mathrm{A,AS}}+\mathbf{R}_{\mathrm{A}}\herm\mathbf{M}_{1}\herm,\label{eq:M3}\\
 & \!\!\!\!\mathbf{M}_{4\mathrm{c}}=\big(\mathbf{W}_{\mathrm{c}}\conj\otimes\mathbf{I}_{m_{\mathrm{S}}}\big)\big(\mathbf{M}_{2}\widetilde{\mathbf{X}}\widehat{\bm{\Sigma}}_{\mathrm{A,AS}}-\mathbf{M}_{3}\big),\label{eq:M4c}\\
 & \!\!\!\!\mathbf{M}_{5\mathrm{c}}=\big(\mathbf{W}_{\mathrm{c}}\conj\otimes\mathbf{I}_{m_{\mathrm{S}}}\big)\mathbf{M}_{2}\breve{\mathbf{X}}\widehat{\bm{\Sigma}}_{\mathrm{A,RS}}\big\{\big(\bm{\Theta}\conj\mathbf{H}_{\mathrm{AR}}\conj\big)\otimes\mathbf{I}_{m_{\mathrm{S}}}\big\},\!\!\label{eq:M5c}\\
 & \!\!\!\!\mathbf{M}_{\mathrm{c}}=\mathbf{M}_{4\mathrm{c}}+\mathbf{M}_{5\mathrm{c}},\label{eq:Mc}\\
 & \!\!\!\!\mathbf{G}_{\mathrm{c}}=(\mathbf{C}_{m_{\mathrm{S}}m_{\min}}\otimes\mathbf{I}_{m_{\mathrm{S}}})(\mathbf{I}_{m_{\min}}\otimes\vecop(\mathbf{I}_{m_{\mathrm{S}}})),\label{eq:Gc}
\end{align}
\end{subequations}
and $\mathbf{C}_{m_{\mathrm{X}}m_{\mathrm{Y}}}\in\mathbb{R}^{m_{\mathrm{X}}m_{\mathrm{Y}}\times m_{\mathrm{X}}m_{\mathrm{Y}}}$ is the commutation matrix. Similarly, the expression for $\op{d}\tr\big\{\mathbf{R}_{\mathrm{A}}(\widehat{\bm{\Sigma}}_{\mathrm{A,ARS}}+\sigma_{\mathrm{S}}^{2}\mathbf{I})\mathbf{R}_{\mathrm{A}}\herm\big\}$ is obtained as 
\begin{align}
 & \op{d}\tr\big\{\mathbf{R}_{\mathrm{A}}(\widehat{\bm{\Sigma}}_{\mathrm{A,ARS}}+\sigma_{\mathrm{S}}^{2}\mathbf{I})\mathbf{R}_{\mathrm{A}}\herm\big\}\nonumber \\
= & \tr\big\{\mathbf{U}_{1}\big(\op{d}\mathbf{R}_{\mathrm{A}}\big)\big\}+\tr\big\{\mathbf{U}_{2}\big(\op{d}\widehat{\bm{\Sigma}}_{\mathrm{A,ARS}}\big)\big\}+\tr\big\{\mathbf{U}_{1}\herm\big(\op{d}\mathbf{R}_{\mathrm{A}}\herm\big)\big\}\nonumber \\
= & \tr\big\{\mathbf{R}_{\mathrm{A}2}^{-1}\mathbf{U}_{1}\widehat{\bm{\Sigma}}_{\mathrm{A,AS}}\big(\op{d}\widetilde{\mathbf{X}}\herm\big)\big\}\nonumber \\
 & \qquad\qquad+\tr\big\{\mathbf{U}_{2}\breve{\mathbf{X}}\widehat{\bm{\Sigma}}_{\mathrm{A,RS}}\big(\op{d}\breve{\mathbf{X}}\herm\big)\big\}-\tr\big\{\mathbf{U}_{3}\big(\op{d}\mathbf{R}_{\mathrm{A}2}\big)\big\}\nonumber \\
= & \tr\big\{\mathbf{U}_{4}\big(\op{d}\widetilde{\mathbf{X}}\herm\big)\big\}+\tr\big\{\mathbf{U}_{5}\big(\op{d}\breve{\mathbf{X}}\herm\big)\big\}\nonumber \\
= & \vecop\trans\big(\mathbf{U}_{\mathrm{c}}\trans\big)\op{d}\big\{\vecop\big(\mathbf{F}_{\mathrm{c}}\conj\otimes\mathbf{I}_{m_{\mathrm{S}}}\big)\big\}\nonumber \\
= & \vecop\trans\big(\mathbf{U}_{\mathrm{c}}\trans\big)\big(\mathbf{I}_{m_{\mathrm{A}}}\otimes\mathbf{G}_{\mathrm{c}}\big)\op{d}\big\{\vecop\big(\mathbf{F}_{\mathrm{c}}\conj\big)\big\},\label{eq:A-4}
\end{align}
where 
\begin{subequations}
\label{eq:def_U_comm}
\begin{align}
 & \mathbf{U}_{1}=\big(\widehat{\bm{\Sigma}}_{\mathrm{A,ARS}}+\sigma_{\mathrm{S}}^{2}\mathbf{I}\big)\mathbf{R}_{\mathrm{A}}\herm,\label{eq:U1}\\
 & \mathbf{U}_{2}=\mathbf{R}_{\mathrm{A}}\herm\mathbf{R}_{\mathrm{A}},\label{eq:U2}\\
 & \mathbf{U}_{3}=\mathbf{R}_{\mathrm{A}2}^{-1}\mathbf{U}_{1}\mathbf{R}_{\mathrm{A}}+\mathbf{R}_{\mathrm{A}}\herm\mathbf{U}_{1}\herm\mathbf{R}_{\mathrm{A}2}^{-1},\label{eq:U3}\\
 & \mathbf{U}_{4}=\big(\mathbf{R}_{\mathrm{A}2}^{-1}\mathbf{U}_{1}-\mathbf{U}_{3}\widetilde{\mathbf{X}}\big)\widehat{\bm{\Sigma}}_{\mathrm{A,AS}},\label{eq:U4}\\
 & \mathbf{U}_{5}=\big(\mathbf{U}_{2}-\mathbf{U}_{3}\big)\breve{\mathbf{X}}\widehat{\bm{\Sigma}}_{\mathrm{A,RS}},\label{eq:U5}\\
 & \mathbf{U}_{\mathrm{c}}=\big(\mathbf{W}_{\mathrm{c}}\conj\otimes\mathbf{I}_{m_{\mathrm{S}}}\big)\big[\mathbf{U}_{4}+\mathbf{U}_{5}\big\{\big(\bm{\Theta}\conj\mathbf{H}_{\mathrm{AR}}\conj\big)\otimes\mathbf{I}_{m_{\mathrm{S}}}\big\}\big].\label{eq:Uc}
\end{align}
\end{subequations}
Therefore, using~\eqref{eq:A-2}, \eqref{eq:A-3} and~\eqref{eq:A-4}, one can write 
\begin{multline}
\!\!\nabla_{\mathbf{F}_{\mathrm{c}}}\bar{\xi}_{\mathrm{pred}}(\mathbf{F},\bm{\theta})=\frac{1}{\tr\big(\bm{\Sigma}_{\mathrm{AS}}\big)}\\
\!\!\times\unvec_{m_{\mathrm{A}}\times m_{\min}}\big(\big[\vecop\trans\big(\mathbf{M}_{\mathrm{c}}\trans+\mathbf{U}_{\mathrm{c}}\trans\big)\big(\mathbf{I}_{m_{\mathrm{A}}}\otimes\mathbf{G}_{\mathrm{c}}\big)\big]\trans\big).\!\!\label{eq:A-5}
\end{multline}
Next, we obtain $\op{d}f(\mathbf{F},\boldsymbol{\theta},\mathscr{C}_{\mathrm{AB}},\tau)$ as follows: 
\begin{align}
 & \op{d}f(\mathbf{F},\boldsymbol{\theta},\mathscr{C}_{\mathrm{AB}},\tau)-\frac{1}{\mathscr{C}_{\mathrm{AB}}}\op{d}C_{\mathrm{AB}}(\mathbf{F},\boldsymbol{\theta})\nonumber \\
= & -\frac{1}{\mathscr{C}_{\mathrm{AB}}}\op{d}\big\{\ln\det\big(\mathbf{I}+\widehat{\mathbf{Z}}_{\mathrm{AB}}\mathbf{F}_{\mathrm{c}}\mathbf{F}_{\mathrm{c}}\herm\widehat{\mathbf{Z}}_{\mathrm{AB}}\herm\mathbf{Q}_{\mathrm{AB}}^{-1}\big)\big\}\nonumber \\
= & \frac{1}{\mathscr{C}_{\mathrm{AB}}}\op{d}\big[\ln\det\big(\mathbf{Q}_{\mathrm{AB}}\big)-\ln\det\big(\mathbf{E}_{\mathrm{AB}}\big)\big]\nonumber \\
= & \frac{1}{\mathscr{C}_{\mathrm{AB}}}\big[\tr\big(\mathbf{Q}_{\mathrm{AB}}^{-1}\op{d}\mathbf{Q}_{\mathrm{AB}}\big)-\tr\big(\mathbf{E}_{\mathrm{AB}}^{-1}\op{d}\mathbf{E}_{\mathrm{AB}}\big)\big]\nonumber \\
= & \frac{1}{\mathscr{C}_{\mathrm{AB}}}\big[\tr\big(\big\{\mathbf{Q}_{\mathrm{AB}}^{-1}-\mathbf{E}_{\mathrm{AB}}^{-1}\big\}\op{d}\mathbf{Q}_{\mathrm{AB}}\big)\nonumber \\
 & \qquad\qquad\qquad\qquad-\tr\big(\widehat{\mathbf{Z}}_{\mathrm{AB}}\herm\mathbf{E}_{\mathrm{AB}}^{-1}\widehat{\mathbf{Z}}_{\mathrm{AB}}\mathbf{F}_{\mathrm{c}}\op{d}\mathbf{F}_{\mathrm{c}}\herm\big)\big].\label{eq:A-6}
\end{align}
where $\mathbf{E}_{\mathrm{AB}}=\mathbf{Q}_{\mathrm{AB}}+\widehat{\mathbf{Z}}_{\mathrm{AB}}\mathbf{F}_{\mathrm{c}}\mathbf{F}_{\mathrm{c}}\herm\widehat{\mathbf{Z}}_{\mathrm{AB}}\herm$. Then using~\eqref{eq:A-6} and~\cite[Table III]{07-TSP_Differential}, one can obtain
\begin{multline}
\nabla_{\mathbf{F}_{\mathrm{c}}}f(\mathbf{F},\boldsymbol{\theta},\mathscr{C}_{\mathrm{AB}},\tau)=\frac{1}{\mathscr{C}_{\mathrm{AB}}}\big[\tr\big(\mathbf{D}_{\mathrm{AB}}\big)\big(\varsigma_{\mathrm{AB}}^{2}\mathbf{I}_{m_{\mathrm{A}}}\\
+\varsigma_{\mathrm{RB}}^{2}\mathbf{H}_{\mathrm{AR}}\herm\mathbf{H}_{\mathrm{AR}}\big)-\widehat{\mathbf{Z}}_{\mathrm{AB}}\herm\mathbf{E}_{\mathrm{AB}}^{-1}\widehat{\mathbf{Z}}_{\mathrm{AB}}\big]\mathbf{F}_{\mathrm{c}},\label{eq:A-7}
\end{multline}
where $\mathbf{D}_{\mathrm{AB}}=\mathbf{Q}_{\mathrm{AB}}^{-1}-\mathbf{E}_{\mathrm{AB}}^{-1}$. Using~\eqref{eq:A-1}, \eqref{eq:A-5} and~\eqref{eq:A-7}, a closed-form expression for $\nabla_{\mathbf{F}_{\mathrm{c}}}g_{\nu,\rho}(\mathbf{F},\boldsymbol{\theta},\tau)$ is given by~\eqref{eq:grad_Fc_Closed}. 

Analogous to~\eqref{eq:A-5}, we obtain the expression for $\nabla_{\mathbf{F}_{\mathrm{s}}}\bar{\xi}_{\mathrm{pred}}(\mathbf{F},\bm{\theta})$ as 
\begin{multline}
\!\!\nabla_{\mathbf{F}_{\mathrm{s}}}\bar{\xi}_{\mathrm{pred}}(\mathbf{F},\bm{\theta})=\frac{1}{\tr\big(\bm{\Sigma}_{\mathrm{AS}}\big)}\\
\!\!\times\unvec_{m_{\mathrm{A}}\times m_{\mathrm{A}}}\big(\big[\vecop\trans\big(\mathbf{M}_{\mathrm{s}}\trans+\mathbf{U}_{\mathrm{s}}\trans\big)\big(\mathbf{I}_{m_{\mathrm{A}}}\otimes\mathbf{G}_{\mathrm{s}}\big)\big]\trans\big),\!\!\label{eq:A-8}
\end{multline}
where 
\begin{subequations}
\label{eq:def_M_sense}
\begin{align}
 & \!\!\!\!\mathbf{M}_{\mathrm{s}}=\mathbf{M}_{4\mathrm{s}}+\mathbf{M}_{5\mathrm{s}},\label{eq:Ms}\\
 & \!\!\!\!\mathbf{M}_{4\mathrm{s}}=\big(\mathbf{W}_{\mathrm{s}}\conj\otimes\mathbf{I}_{m_{\mathrm{S}}}\big)\big(\mathbf{M}_{2}\widetilde{\mathbf{X}}\widehat{\bm{\Sigma}}_{\mathrm{A,AS}}-\mathbf{M}_{3}\big),\label{eq:M4s}\\
 & \!\!\!\!\mathbf{M}_{5\mathrm{s}}=\big(\mathbf{W}_{\mathrm{s}}\conj\otimes\mathbf{I}_{m_{\mathrm{S}}}\big)\mathbf{M}_{2}\breve{\mathbf{X}}\widehat{\bm{\Sigma}}_{\mathrm{A,RS}}\big\{\big(\bm{\Theta}\conj\mathbf{H}_{\mathrm{AR}}\conj\big)\otimes\mathbf{I}_{m_{\mathrm{S}}}\big\},\!\!\label{eq:M5s}\\
 & \!\!\!\!\mathbf{G}_{\mathrm{s}}=(\mathbf{C}_{m_{\mathrm{S}}m_{\mathrm{A}}}\otimes\mathbf{I}_{m_{\mathrm{S}}})(\mathbf{I}_{m_{\mathrm{A}}}\otimes\vecop(\mathbf{I}_{m_{\mathrm{S}}})),\label{eq:Gs}\\
 & \!\!\!\!\mathbf{U}_{\mathrm{s}}=\big(\mathbf{W}_{\mathrm{s}}\conj\otimes\mathbf{I}_{m_{\mathrm{S}}}\big)\big[\mathbf{U}_{4}+\mathbf{U}_{5}\big\{\big(\bm{\Theta}\conj\mathbf{H}_{\mathrm{AR}}\conj\big)\otimes\mathbf{I}_{m_{\mathrm{S}}}\big\}\big].\label{eq:Us}
\end{align}
\end{subequations}
 Similarly, following~\eqref{eq:A-7}, we obtain 
\begin{multline}
\nabla_{\mathbf{F}_{\mathrm{s}}}f(\mathbf{F},\boldsymbol{\theta},\mathscr{C}_{\mathrm{AB}},\tau)=\frac{1}{\mathscr{C}_{\mathrm{AB}}}\tr\big(\mathbf{D}_{\mathrm{AB}}\big)\\
\times\big(\widehat{\mathbf{Z}}_{\mathrm{AB}}\herm\widehat{\mathbf{Z}}_{\mathrm{AB}}+\varsigma_{\mathrm{AB}}^{2}\mathbf{I}_{m_{\mathrm{A}}}+\varsigma_{\mathrm{RB}}^{2}\mathbf{H}_{\mathrm{AR}}\herm\mathbf{H}_{\mathrm{AR}}\big)\mathbf{F}_{\mathrm{s}}.\label{eq:A-9}
\end{multline}
Using~\eqref{eq:A-8} and~\eqref{eq:A-9}, a closed-form expression for $\nabla_{\mathbf{F}_{\mathrm{s}}}g_{\nu,\rho}(\mathbf{F},\boldsymbol{\theta},\tau)$ is given by~\eqref{eq:grad_Fs_Closed}. This completes the proof. 

\section{\label{sec:proof_grad_theta}Proof of Theorem~\ref{thm:grad_theta}}

When $\mathbf{F}$ and $\tau$ are held fixed, the complex-valued differential of $g_{\nu,\rho}(\mathbf{F},\boldsymbol{\theta},\tau)$ w.r.t. $\bm{\theta}$ can be obtained as 
\begin{multline}
\op{d}g_{\nu,\rho}(\mathbf{F},\boldsymbol{\theta},\tau)=\op{d}\bar{\xi}_{\mathrm{pred}}(\mathbf{F},\bm{\theta})\\
-\big\{\nu+\tfrac{1}{\rho}f(\mathbf{F},\boldsymbol{\theta},\mathscr{C}_{\mathrm{AB}},\tau)\big\}\op{d}f(\mathbf{F},\boldsymbol{\theta},\mathscr{C}_{\mathrm{AB}},\tau).\label{eq:B-1}
\end{multline}
For the first term on the RHS, we have 
\begin{align}
 & \op{d}\bar{\xi}_{\mathrm{pred}}(\mathbf{F},\bm{\theta})\nonumber \\
=\  & \op{d}\frac{1}{\tr\big(\bm{\Sigma}_{\mathrm{AS}}\big)}\Big[\tr\big\{(\mathbf{I}-\mathbf{R}_{\mathrm{A}}\widetilde{\mathbf{X}})\widehat{\bm{\Sigma}}_{\mathrm{A,AS}}(\mathbf{I}-\mathbf{R}_{\mathrm{A}}\widetilde{\mathbf{X}})\herm\big\}\nonumber \\
 & \qquad\qquad\qquad\qquad+\tr\big\{\mathbf{R}_{\mathrm{A}}(\widehat{\bm{\Sigma}}_{\mathrm{A,ARS}}+\sigma_{\mathrm{S}}^{2}\mathbf{I})\mathbf{R}_{\mathrm{A}}\herm\big\}\Big]\nonumber \\
=\  & \frac{1}{\tr\big(\bm{\Sigma}_{\mathrm{AS}}\big)}\Big[\op{d}\tr\big\{(\mathbf{I}-\mathbf{R}_{\mathrm{A}}\widetilde{\mathbf{X}})\widehat{\bm{\Sigma}}_{\mathrm{A,AS}}(\mathbf{I}-\mathbf{R}_{\mathrm{A}}\widetilde{\mathbf{X}})\herm\big\}\nonumber \\
 & \qquad\qquad\qquad+\op{d}\tr\big\{\mathbf{R}_{\mathrm{A}}(\widehat{\bm{\Sigma}}_{\mathrm{A,ARS}}+\sigma_{\mathrm{S}}^{2}\mathbf{I})\mathbf{R}_{\mathrm{A}}\herm\big\}\Big].\label{eq:B-2}
\end{align}
The expression for $\op{d}\tr\big\{(\mathbf{I}-\mathbf{R}_{\mathrm{A}}\widetilde{\mathbf{X}})\widehat{\bm{\Sigma}}_{\mathrm{A,AS}}(\mathbf{I}-\mathbf{R}_{\mathrm{A}}\widetilde{\mathbf{X}})\herm\big\}$ can be obtained as 
\begin{align}
 & \op{d}\tr\big\{(\mathbf{I}-\mathbf{R}_{\mathrm{A}}\widetilde{\mathbf{X}})\widehat{\bm{\Sigma}}_{\mathrm{A,AS}}(\mathbf{I}-\mathbf{R}_{\mathrm{A}}\widetilde{\mathbf{X}})\herm\big\}\nonumber \\
=\  & \tr\big\{\op{d}(\mathbf{I}-\mathbf{R}_{\mathrm{A}}\widetilde{\mathbf{X}})\mathbf{M}_{1}\big\}+\tr\big\{\mathbf{M}_{1}\herm\op{d}(\mathbf{I}-\mathbf{R}_{\mathrm{A}}\widetilde{\mathbf{X}})\herm\big\}\nonumber \\
=\  & \tr\big\{\mathbf{R}_{\mathrm{A}2}^{-1}\mathbf{V}_{1}\mathbf{R}_{\mathrm{A}2}^{-1}\big(\op{d}\mathbf{R}_{\mathrm{A}2}\big)\big\}\nonumber \\
=\  & \tr\big\{\mathbf{V}_{2}\op{d}\big[\big(\bm{\Theta}\conj\mathbf{H}_{\mathrm{AR}}\conj\mathbf{X}\conj\big)\otimes\mathbf{I}_{m_{\mathrm{S}}}\big]\big\} = \tr\big\{\mathbf{V}_{3}\op{d}\big(\bm{\Theta}\conj\otimes\mathbf{I}_{m_{\mathrm{S}}}\big)\big\}\nonumber \\
=\  & \vecop\trans\big(\mathbf{V}_{3}\trans\big)\big(\mathbf{I}_{m_{\mathrm{R}}}\otimes\mathbf{G}_{\theta}\big)\op{d}\big\{\vecop\big(\bm{\Theta}\conj\big)\big\}.\label{eq:B-3}
\end{align}
where $\mathbf{M}_{1}$ is given in~\eqref{eq:M1}, $\mathbf{V}_{1}=\widetilde{\mathbf{X}}\mathbf{M}_{1}\mathbf{R}_{\mathrm{A}1}+\big(\widetilde{\mathbf{X}}\mathbf{M}_{1}\mathbf{R}_{\mathrm{A}1}\big)\herm$, $\mathbf{V}_{2}=\mathbf{R}_{\mathrm{A}2}^{-1}\mathbf{V}_{1}\mathbf{R}_{\mathrm{A}2}^{-1}\breve{\mathbf{X}}\widehat{\bm{\Sigma}}_{\mathrm{A,RS}}$, $\mathbf{V}_{3}=\big[\big(\mathbf{H}_{\mathrm{AR}}\conj\mathbf{X}\conj\big)\otimes\mathbf{I}_{m_{\mathrm{S}}}\big]\mathbf{V}_{2}$, $\mathbf{G}_{\theta}=\big(\mathbf{C}_{m_{\mathrm{R}}m_{\mathrm{S}}}\otimes\mathbf{I}_{m_{\mathrm{S}}}\big)\big(\mathbf{I}_{m_{\mathrm{R}}}\otimes\vecop\big(\mathbf{I}_{m_{\mathrm{S}}}\big)\big)$, and 
$\mathbf{C}_{m_{\mathrm{R}}m_{\mathrm{S}}}$ is the commutation matrix. Similarly, for $\op{d}\tr\big\{\mathbf{R}_{\mathrm{A}}(\widehat{\bm{\Sigma}}_{\mathrm{A,ARS}}+\sigma_{\mathrm{S}}^{2}\mathbf{I})\mathbf{R}_{\mathrm{A}}\herm\big\}$, we have 
\begin{align}
 & \op{d}\tr\big\{\mathbf{R}_{\mathrm{A}}(\widehat{\bm{\Sigma}}_{\mathrm{A,ARS}}+\sigma_{\mathrm{S}}^{2}\mathbf{I})\mathbf{R}_{\mathrm{A}}\herm\big\}\nonumber \\
=\  & \tr\big\{\mathbf{U}_{1}\big(\op{d}\mathbf{R}_{\mathrm{A}}\big)\big\}\!+\!\tr\big\{\mathbf{U}_{2}\big(\op{d}\widehat{\bm{\Sigma}}_{\mathrm{A,ARS}}\big)\big\}\!+\!\tr\big\{\mathbf{U}_{1}\herm\big(\op{d}\mathbf{R}_{\mathrm{A}}\herm\big)\big\}\nonumber \\
=\  & \tr\big\{\mathbf{J}_{1}\big(\op{d}\mathbf{R}_{\mathrm{A}2}^{-1}\big)\big\}+\tr\big\{\mathbf{J}_{2}\big(\op{d}\breve{\mathbf{X}}\herm\big)\big\} = \tr\big\{\mathbf{J}_{3}\big(\op{d}\breve{\mathbf{X}}\herm\big)\big\} \nonumber \\
=\  & \tr\big\{\mathbf{J}_{4}\op{d}\big(\bm{\Theta}\conj\otimes\mathbf{I}_{m_{\mathrm{S}}}\big)\big\} = \vecop\trans\big(\mathbf{J}_{4}\trans\big)\op{d}\big\{\vecop\big(\bm{\Theta}\conj\otimes\mathbf{I}_{m_{\mathrm{S}}}\big)\big\} \nonumber \\
=\  & \vecop\trans\big(\mathbf{J}_{4}\trans\big)\big(\mathbf{I}_{m_{\mathrm{R}}}\otimes\mathbf{G}_{\theta}\big)\op{d}\big\{\vecop\big(\bm{\Theta}\conj\big)\big\},\label{eq:B-4}
\end{align}
where $\mathbf{U}_{1}$ is given in~\eqref{eq:U1}, $\mathbf{U}_{2}$ is given in~\eqref{eq:U2}, $\mathbf{J}_{1}=\mathbf{U}_{1}\mathbf{R}_{\mathrm{A}1}+\big(\mathbf{U}_{1}\mathbf{R}_{\mathrm{A}1}\big)\herm$, $\mathbf{J}_{2}=\mathbf{U}_{2}\breve{\mathbf{X}}\widehat{\bm{\Sigma}}_{\mathrm{A,RS}}$, $\mathbf{J}_{3}=\mathbf{J}_{2}-\mathbf{R}_{\mathrm{A2}}^{-1}\mathbf{J}_{1}\mathbf{R}_{\mathrm{A}2}^{-1}\breve{\mathbf{X}}\widehat{\bm{\Sigma}}_{\mathrm{A,RS}}$, and $\mathbf{J}_{4}=\big\{\big(\mathbf{H}_{\mathrm{AR}}\conj\mathbf{X}\conj\big)\otimes\mathbf{I}_{m_{\mathrm{S}}}\big\}\mathbf{J}_{3}$. 
Therefore, using~\eqref{eq:B-2}, \eqref{eq:B-3} and~\eqref{eq:B-4}, a closed-form expression for $\nabla_{\bm{\theta}}\bar{\xi}_{\mathrm{pred}}(\mathbf{F},\bm{\theta})$ is given by 
\begin{multline}
\nabla_{\bm{\theta}}\bar{\xi}_{\mathrm{pred}}(\mathbf{F},\bm{\theta})=\vecd\Big[\frac{1}{\tr\big(\bm{\Sigma}_{\mathrm{AS}}\big)}\\
\times\unvec_{m_{\mathrm{R}}\times m_{\mathrm{R}}}\big(\big[\vecop\trans\big(\mathbf{V}_{3}\trans+\mathbf{J}_{4}\trans\big)\big(\mathbf{I}_{m_{\mathrm{R}}}\otimes\mathbf{G}_{\theta}\big)\big]\trans\big)\Big].\!\!\!\!\label{eq:B-5}
\end{multline}
Moreover, a closed-form expression for $\op{d}f(\mathbf{F},\boldsymbol{\theta},\mathscr{C}_{\mathrm{AB}},\tau)$ is obtained as follows:
\begin{align}
 & \op{d}f(\mathbf{F},\boldsymbol{\theta},\mathscr{C}_{\mathrm{AB}},\tau)=-\frac{1}{\mathscr{C}_{\mathrm{AB}}}\operatorname{d}C_{\mathrm{AB}}(\mathbf{F},\boldsymbol{\theta})\nonumber \\
=\  & \frac{1}{\mathscr{C}_{\mathrm{AB}}}\big[\tr\big(\mathbf{Q}_{\mathrm{AB}}^{-1}\op{d}\mathbf{Q}_{\mathrm{AB}}\big)-\tr\big(\mathbf{E}_{\mathrm{AB}}^{-1}\op{d}\mathbf{E}_{\mathrm{AB}}\big)\big]\nonumber \\
=\  & \frac{1}{\mathscr{C}_{\mathrm{AB}}}\big[\tr\big(\big\{\mathbf{D}_{\mathrm{AB}}\big\}\op{d}\mathbf{Q}_{\mathrm{AB}}\big)-\tr\big(\mathbf{E}_{\mathrm{AB}}^{-1}\widehat{\mathbf{Z}}_{\mathrm{AB}}\mathbf{F}_{\mathrm{c}}\mathbf{F}_{\mathrm{c}}\herm\op{d}\widehat{\mathbf{Z}}_{\mathrm{AB}}\herm\big)\big]\nonumber \\
=\  & \frac{1}{\mathscr{C}_{\mathrm{AB}}}\tr\big\{\mathbf{L}_{\mathrm{AB}}\op{d}\widehat{\mathbf{Z}}_{\mathrm{AB}}\herm\big\}=\frac{1}{\mathscr{C}_{\mathrm{AB}}}\tr\big\{\widehat{\mathbf{H}}_{\mathrm{RB}}\herm\mathbf{L}_{\mathrm{AB}}\mathbf{H}_{\mathrm{AR}}\herm\op{d}\bm{\Theta}\herm\big\},\label{eq:B-6}
\end{align}
where $\mathbf{E}_{\mathrm{AB}}$ and $\mathbf{D}_{\mathrm{AB}}$ are given in~Appendix~\ref{sec:proof_grad_Fc}, and $\mathbf{L}_{\mathrm{AB}}=\mathbf{D}_{\mathrm{AB}}\widehat{\mathbf{Z}}_{\mathrm{AB}}\mathbf{F}_{\mathrm{s}}\mathbf{F}_{\mathrm{s}}\herm-\mathbf{E}_{\mathrm{AB}}^{-1}\widehat{\mathbf{Z}}_{\mathrm{AB}}\mathbf{F}_{\mathrm{c}}\mathbf{F}_{\mathrm{c}}\herm$. Hence, one can obtain $\nabla_{\bm{\theta}}f(\mathbf{F},\boldsymbol{\theta},\mathscr{C}_{\mathrm{AB}},\tau)$ is given by 
\begin{equation}
\nabla_{\bm{\theta}}f(\mathbf{F},\boldsymbol{\theta},\mathscr{C}_{\mathrm{AB}},\tau)=\frac{1}{\mathscr{C}_{\mathrm{AB}}}\vecd\big(\widehat{\mathbf{H}}_{\mathrm{RB}}\herm\mathbf{L}_{\mathrm{AB}}\mathbf{H}_{\mathrm{AR}}\herm\big).\label{eq:B-7}
\end{equation}
Using~\eqref{eq:B-5} and~\eqref{eq:B-7}, a closed-form expression for $\nabla_{\bm{\theta}}g_{\nu,\rho}(\mathbf{F},\boldsymbol{\theta},\tau)$ is given by~\eqref{eq:grad_theta_Closed}; this concludes the proof. 

\bibliographystyle{IEEEtran}
\bibliography{references}

\end{document}